\documentclass[11pt]{article}

\usepackage{amsmath,fullpage}
\usepackage{amssymb}
\usepackage{algorithm}
\usepackage{algpseudocode}
\usepackage[latin1]{inputenc} 
\usepackage[T1]{fontenc}
\usepackage{graphicx}
\usepackage{graphics}
\usepackage{stmaryrd}
\usepackage{tikz}
\usetikzlibrary{decorations.pathreplacing,snakes,arrows}
\usepackage{caption}
\usepackage{verbatim}
\usepackage[raggedright,nooneline,normalsize]{subfigure}
\usepackage{enumerate}
\usepackage{array}
\usepackage{nicefrac}
\usepackage{multirow}
\usepackage{paralist}
\usepackage{varwidth}
\usepackage{amsthm}

\newenvironment{theorem*}{{\bf Theorem. }}{}
\newenvironment{lemma*}{{\bf Lemma. }}{}
\newtheorem{theorem}{Theorem}
\newtheorem{lemma}[theorem]{Lemma}

\newtheorem{corollary}[theorem]{Corollary}
\newtheorem{proposition}[theorem]{Proposition}

\newcommand{\PFG}{Progressive Filling Game}
\newcommand{\pfg}{progressive filling game}

\DeclareMathOperator{\SW}{SW}

\DeclareMathOperator{\PoS}{PoS}

\DeclareMathOperator{\ord}{ord}
\DeclareMathOperator*{\argmin}{argmin}
\DeclareMathOperator*{\argmax}{argmax}
\DeclareMathOperator{\img}{img}

\begin{document}

\title{Routing Games with Progressive Filling\thanks{This research was partially supported by the German Research
Foundation (DFG) within the Cluster of Excellence MMCI at Saarland University, the Research Training Group ``Methods for Discrete Structures'' (GRK 1408) and the Collaborative Research Center ``On-The-Fly Computing'' (SFB 901). It was also partially supported by the EU within FET project MULTIPLEX (contract no.\ 317532) and the Marie-Curie grant ``Protocol Design'' (no. 327546, funded within FP7-PEOPLE-2012-IEF).}}

\author{Tobias Harks%
\thanks{Dept.\ of Quantitative Economics, Maastricht University, {\tt t.harks@maastrichtuniversity.nl}}%
\and%
Martin Hoefer%
\thanks{Max-Planck-Institut f\"ur Informatik and Saarland University, {\tt mhoefer@mpi-inf.mpg.de}}%
\and%
Kevin Schewior\thanks{Institut f\"ur Mathematik, TU Berlin, {\tt schewior@math.tu-berlin.de}}%
\and%
Alexander Skopalik\thanks{Dept.\ of Computer Science, Paderborn University, {\tt skopalik@mail.upb.de}}
}

\maketitle


\begin{abstract}
Max-min fairness (MMF) is a widely known approach to a fair allocation of bandwidth to each of the users in a network. This allocation can be computed by uniformly raising the bandwidths of all users without violating capacity constraints. We consider an extension of these allocations by raising the bandwidth with arbitrary and not necessarily uniform time-depending velocities (allocation rates). These allocations are used in a game-theoretic context for routing choices, which we formalize in progressive filling games (PFGs). 

We present a variety of results for equilibria in PFGs. We show that these games possess pure Nash and strong equilibria. While computation in general is NP-hard, there are polynomial-time algorithms for prominent classes of Max-Min-Fair Games (MMFG), including the case when all users have the same source-destination pair. We characterize prices of anarchy and stability for pure Nash and strong equilibria in PFGs and MMFGs when players have different or the same source-destination pairs. In addition, we show that when a designer can adjust allocation rates, it is possible to design games with optimal strong equilibria. Some initial results on polynomial-time algorithms in this direction are also derived.

\end{abstract}

\section{Introduction}
Max-min fairness is a widely used paradigm for bandwidth allocation
problems in telecommunication networks, most prominently, it is used as a reference
point for designing flow control/congestion control protocols such as TCP (Transport Control Protocol), see~\cite{Tan:2007} for a more detailed discussion.  
In a max-min fair allocation, the bandwidth of a user cannot be increased without decreasing
the bandwidth of another user, who already receives a smaller bandwidth.
Max-min fairness also plays an important role in the model of Kelly et al.~\cite{Kelly98},
where  congestion control protocols have been interpreted as distributed algorithms at sources and links in
order to solve a global optimization
problem (cf.~\cite{LowLapsley99,LowPagDoy2002,Mo:2000} for further works in this area).  Each
user is associated with an increasing, strictly concave bandwidth
utility function and the congestion control
algorithms aim at maximizing aggregate utility subject to capacity
constraints on the links.  Mo and Walrand~\cite{Mo:2000} showed that
within the model of Kelly et al., there is a family of utility functions
whose global optimum corresponds to a max-min fair bandwidth allocation
and they devised a distributed max-min fair congestion control protocol, see also~\cite{Srikant03} (Section 2.2) and \cite{Mamatas07}. For further 
distributed max-min fair congestion control
protocols, we refer to~\cite{wydr07,zhang07}.
There are several important generalizations of
max-min fairness such as weighted max-min fairness~\cite{wydr07}
and utility max-min fairness~\cite{HP08}.
In a weighted max-min fair allocation, the weighted bandwidth of a user cannot be increased without decreasing
the weighted bandwidth of another user, who already receives a smaller weighted bandwidth.
In a utility max-min fair allocation, each user is associated with an increasing
(not necessarily concave) bandwidth utility function and an allocation is utility max-min fair if the utility of a user cannot be increased without decreasing
the utility of another user, who already receives a smaller utility.
Utility max-min fairness (and also weighted max-min fairness) has been proposed for 
giving some applications (e.g., real-time applications, or multi-media) a possibly larger bandwidth share  than others. 

It is well known that (weighted) max-min fair allocations can be easily implemented by simple polynomial time water-filling algorithms that raise the bandwidth of every user at a (weighted) uniform speed
and, whenever a link capacity is exhausted, fixes the bandwidth of those users traversing this link~\cite{BertsekasG92}. As we will show in this paper, also
utility max-min fair allocations can be implemented by simple polynomial time water-filling algorithms that raise the bandwidth of every user at a user-specific speed.


While most works in the area of flow control/congestion control assume
that the routes of users are fixed a priori,
we study in this paper the flexibility of \emph{strategic route choices} by users (or players from now on) as a means 
to obtain high bandwidth. We introduce a general class of strategic games that we term
\emph{routing games with progressive filling}. In such a game, there is a finite set of resources 
and a strategy of a player corresponds to a subset of resources. Resources have capacities and the
utility of every player equals the obtained bandwidth which in turn is defined by a predefined 
water-filling algorithm. If the allowable subsets of a player correspond to the set of routes
connecting the player's source with its terminal, we obtain single-path routing modeling IP 
(Internet Protocol) routing. Since IP routing is typically updated at a much slower timescale 
than the flow control, we assume that flow 
control (modeled in this paper as a water-filling algorithm) converges instantly to a "fair" 
allocation (max-min fair or generalizations thereof) after each route update.
The assumption that flow control converges instantly before route updates are triggered has been made and justified before, see, e.g., Wang et al.~\cite{wang2005}. Thus, once a player 
chooses a new route his bandwidth share is determined by executing the water-filling algorithm.
We will impose mild conditions on the class of allowable water-filling algorithms:
(i) for every player and every point in time the integral of the rate function is non-negative
and the integral of the rate function grows monotonically;
(ii) for every player the integral of the rate function tends to infinity as time goes to infinity.
While condition (i) is natural, condition (ii) simply ensures that the water filling
algorithm terminates and the induced strategic game is well-defined.
 Note that even though water-filling algorithms are \emph{centralized} algorithms we demonstrate
 that they represent a wide range of fairness concepts including max-min fairness, weighted 
max-min fairness and utility max-min fairness for which \emph{distributed} and \emph{fast converging} congestion control
protocols are known~\cite{MH08,Mo:2000,wydr07,zhang07}.

We consider existence, computation and quality of equilibria in routing games with progressive filling. 
In a pure Nash equilibrium (PNE for short), no player obtains strictly higher bandwidth by unilaterally 
changing his route. If coordinated deviations by players are allowed (for instance by a single player 
coordinating several sessions or by a set of  players connected via peer-to-peer overlay networks), the 
Nash equilibrium concept is not sufficient to analyze stable states of a game. For this situation, we 
adopt the stricter notion of a strong equilibrium (SE for short) proposed by Aumann~\cite{Aumann59}. 
In a SE, no coalition (of any size) can change their routes and strictly increase the bandwidth of each of its members (while possibly lowering the bandwidth of players outside the coalition). Every SE is a PNE, but not conversely. Thus, SE constitute a very robust and appealing stability concept for which only a few existence results are known
in the literature.

\subsection{Our Results}
\paragraph{Existence.}
For progressive filling games we prove that if water-filling algorithms satisfy conditions 
(i) and (ii), every sequence of profitable deviations of coalitions of players must be finite and, 
hence, SE always exist. Previously, it was only known that PNE exist if the water-filling 
algorithm corresponds to the max-min fair allocation~\cite{YangXFMZ10}.
Thus, our results establish for the first time that routing and congestion control
admits a PNE (and even SE) for routing games where 
weighted- and utility max-min fair congestion control protocols are used.
 We show that our assumptions (i) and 
(ii) are "minimal" in the sense that if one of them is dropped, there is a corresponding two-player 
game without PNE.

\paragraph{Complexity.}
In light of its practical importance, we study routing games with water-filling algorithms inducing the
max-min fair allocation. We first focus on the computational complexity of SE and PNE. We give an 
algorithm that computes a SE for any progressive filling game under max-min fair allocations. 
Our algorithm iteratively reduces the number of players allowed on a resource. 
After each such reduction, a \emph{packing oracle} is invoked that checks whether or not there is a 
feasible strategy profile that respects the allowed numbers of players on every resource.
If the oracle finds a feasible allocation, the algorithm proceeds and, otherwise,
we fix strategies for a suitable subset of players.
Obviously, the running time of the algorithm crucially relies on the running time of the packing oracle.
It is known, however, that if the strategy spaces correspond to, e.g., the set of paths of a
single-commodity network, or to bases of a matroid defined on a player-specific subset of resources,
the oracle can be implemented in polynomial time, thereby ensuring polynomial-time computation of SE.
We complement this result by showing various hardness results of computing SE. In addition, we show 
a bound on the number of values of the potential function that also represents an upper bound on the 
number of improvement steps to reach a PNE.

\paragraph{Quality.}
To measure the quality of an equilibrium, we use the achieved throughput defined as the sum of the 
player's bandwidths. This performance measure corresponds to utilitarian social welfare and is the 
standard performance measure in traffic engineering. We use notions of price of stability (PoS) and 
price of anarchy (PoA), which relate the cost of an equilibrium to the cost of a social optimum. 
The standard definition of an optimum would refer to a set of route choices such that throughput is 
maximized for a waterfilling algorithm with given allocation rates. In addition, our bounds continue 
to hold even with respect to an optimum that is allowed to set arbitrary routes and bandwidths 
respecting the resource capacities. Computing this general optimum is known in combinatorial 
optimization as the maximum k-splittable flow problem.

We provide tight bounds for SE and PNE. In general, the PoS and PoA are $n$, which is tight for both 
PNE and SE, even in single-commodity PFGs or multi-commodity MMFGs. In single-commodity MMFGs, PoS for 
PNE and SE is $\left(2-\frac{1}{n}\right)$, PoA for PNE is $n$ and PoA for SE is $4$. All bounds except 
the latter are tight. In addition, our algorithm that computes SE for single-commodity MMFGs in polynomial 
time yields SE that match the PoS bound. In addition, we show some improved bounds on the PoA for PNE in 
singleton PFGs. 

\paragraph{Protocol Design.}
Using fixed allocation rates, improving upon the $\left(2-\frac{1}{n}\right)$-bound is impossible in the worst case. We show, however, that it is possible to show better results when we have slight flexibility in allocation rates. We assume the freedom to ``design a protocol'' and adjust weights in a weighted MMF waterfilling algorithm towards the topology of the instance. This allows to design a game with an optimal SE that coincides with the maximum k-splittable flow. While computing such an optimum is NP-hard, the result also shows that starting from any $\alpha$-approximation to the maximum k-splittable flow, we can design weights and a starting state, such that every sequence of unilateral (coalitional) improvement moves leads to a PNE (SE) with the same approximation ratio. We apply this approach in games with 3 players, where we can find in polynomial time a solution that is a 1.5-approximation and represents a PNE for the chosen weights.
 
\subsection{Related Work}
Combined routing and congestion control
has been studied by several works (cf. \cite{han06,wang03,key07,voice05}).
In all these works, the existence of an equilibrium is proved by  showing that it corresponds to an optimal solution of an associated convex utility maximization problem.
This, however, implies that every user possibly splits the flow
among an exponential number of routes which might be critical for
some applications. For instance, the standard
TCP/IP protocol suite uses single path routing, because splitting the demand comes with several practical complications, e.g., packets arriving out of order, packet jitter due to different
path delays etc. This issue has been explicitly addressed
by Orda et al.~\cite{Orda93}.

Another related class of games are \textit{congestion games}, where there is a set of resources, and the pure strategies of players are subsets of this set. Each resource has a delay function depending on the \emph{load},
i.e., the number of players that select strategies containing the
respective resource. These games allow to model network structures, but 
they fail to incorporate a realistic allocation of network capacities. The reason is that, even though we can define bandwidths allocated on an edge as a function of the number of players using it, the bandwidth of a player would be given by the \textit{sum} of bandwidth allocated on edges he uses.
This problem is addressed by \textit{bottleneck congestion games}~\cite{ColeDR12} where the bandwidth of one player is rather given by the \textit{maximum} bandwidth among the edges he uses. It is known that strong equilibria exist for bottleneck congestion games~\cite{HarksKM09}. The complexity of computing PNE and SE in these games was further investigated in~\cite{HarksHKS10}, where a central result is an algorithm called Dual Greedy that computes SE. On single-commodity network or matroid bottleneck congestion games, it can be implemented to run in polynomial time. Still, for an arbitrary state, the computation of a coalitional improvement step turned out to be NP-hard, even for these classes. 
The PoA for PNE in bottleneck games can be polynomial in the network size, for social cost being the sum of player delays~\cite{ColeDR12} or maximum player delay~\cite{BuschM09,BannerO07}. For the latter cost function, the PoA for SE becomes 2 for symmetric games with linear delays~\cite{KeijzerST10}. Improved results were obtained for variants, in which players individual costs are exponential or polynomial functions of their delays~\cite{KannanB10,BuschKS2011}. 

A fundamental drawback of bottleneck congestion games is that the bandwidth allocated to a player on a specific edge is \textit{solely} a function of the number of players on it. If one of the players does not exhaust his allocated bandwidth share (e.g., because he has a smaller bottleneck on another edge) the remaining bandwidth remains unused. In max-min fair allocations~\cite{Jaffe81}, this leftover is fairly distributed among players who can make use of it.
  
Yang et al.~\cite{YangXFMZ10} introduced so-called \textit{MAXBAR}-games which correspond to
progressive filling games using max-min fair allocations. They show that these
games possess PNE and that the price of anarchy for PNE is $n$ in
these games, where $n$ is the number of players.
It is also shown that iterative computation of unilateral improvement steps
converge in polynomial time to a PNE if the number of players is constant. 

Amaldi et al.~\cite{AmaldiCCG13} considered a centralized approach to computing routes maximizing the aggregated bandwidth subject to max-min fair allocations. They show hardness results (for multi-commodity networks) and devise an exact algorithm using column generation. Kleinberg et al.~\cite{KleinbergRT01} devise approximation algorithms (and hardness results) for the same problem using an approximate version of max-min fairness.
  
In terms of combinatorial optimization, the problem of computing a strategy
with maximum aggregated bandwidth (without fairness constraints) is related to the \textit{maximum $k$-splittable flow problem}~\cite{BaierKS05}. In contrast to the ordinary maximum flow problem, the number of paths flow is sent along is bounded by $k$, for each commodity. Positive results were especially found for the single-commodity case. For $k=2$ and $k=3$, a $\frac{3}{2}$-approximation was given and this result was generalized to a $2$-approximation for arbitrary fixed $k$. It turned out that, asymptotically, any approximation with a factor of smaller than $\frac{6}{5}$ is NP-hard to obtain. Furthermore, for $k=2$, $\frac{3}{2}$ is exactly the inapproximability bound~\cite{KochS06}.

%

\section{\PFG}
 
A {\em \pfg} \ is a tuple $\left(N,R,\left(c_i\right)_{r\in R},\left(\mathcal{S}_i\right)_{i\in N}, \left(v_i\right)_{i \in N},\left(u_i\right)_{i \in N}\right)$, where
$N=\left\{1,\dots,n\right\}$ is the set of players,
$R=\left\{1,\dots,m\right\}$ is the set of resources,
$c_r\in\mathbb{R}_+$ is the capacity of resource $r$ for each $r\in R$.
The allocation rate is defined as
 $v_i:\mathbb{R}_+\rightarrow\mathbb{R}_+$ and is assumed to be
(Riemann) integrable.
The aggregated rate (or bandwidth) of player $i$ at time $t'$
is defined as $V_i(t')=\int_0^{t'} v_i\left(t\right)dt$.
We assume that for all $i\in N$, $V_i\geq 0$,
$V_i(t) $ is monotonically non-decreasing in $t$, and $\lim_{t\rightarrow\infty}V_i(t)=\infty$.
We denote by
$\mathcal{S}_i\subseteq\mathcal{P}\left(R\right)$ the set of strategies of player $i$, for each $i\in N$, and $\mathcal{S}=\mathcal{S}_1\times\dots\times\mathcal{S}_n$ are the set of states.
Note that this definition is kept very general and can be restricted to model more specific objects, e.g. networks.
An {\em allocation} in state $S\in\mathcal{S}$ is a vector $a=\left(a_1,\dots,a_n\right)\in\mathbb{R}^n$ of feasible bandwidths, i.e.,
$\sum_{i\in N:r\in S_i}a_i\leq c_r,$  for each $r\in R$. The $i$-th component of $a$ is called the bandwidth or capacity of player $i$ (in $a$).
%
Given $S$, we create an allocation the following way. Each of the players starts off with a bandwidth $b_i=0$. We raise their bandwidths with the velocity $v_i\left(t\right)$ at time step $t\in\mathbb{R}$ until a further increase would lead to non-feasible capacities (i.e., one of the resources is \textit{saturated}). At this point, we fix the bandwidths of all the corresponding (\textit{saturated}) players and continue with the other ones. See Algorithm~\ref{alg_pf} for a formal description. For given $S$, we denote by $t_i(S)$  the \emph{ finishing time}, i.e., the time  when player $i$'s bandwidth is fixed. Thus, the payoff for player $i$ is given by $u_i(S)=V_i(t_i(S))=a_i$. We can easily extend our model to allow for
player-specific payoff functions of the form $u_i(S)=U_i(V_i(t_i(S)))=U_i(a_i),$
where $U_i$ is a differentiable and strictly increasing bandwidth utility function.
As long as $U_i$ is strictly increasing (yielding a monotone payoff transformation), an allocation is a PNE (SE) in the new game iff
it is one in the original game. 
We now state a useful observation linking the outcome
of Algorithm~\ref{alg_pf} with different fairness concepts.
\begin{proposition}
Let $U_i,i\in N$ be a set of nonnegative,
differentiable and strictly increasing bandwidth utility functions and
let $w_i,i\in N$ be a set of nonnegative weights. 
For given \pfg\ and state $S$, the following holds:
\begin{enumerate}
\item\label{prop:mmf} If for all $i\in N: v_i(t)=1$ and $u_i(S)=V_i(t_i(S))$, Algorithm~\ref{alg_pf} computes a max-min fair bandwidth allocation under $S$.
\item\label{prop:wmmf} If for all $i\in N: v_i(t)=w_i$ and $u_i(S)=V_i(t_i(S))$, Algorithm~\ref{alg_pf} computes a weighted max-min fair bandwidth allocation under $S$.
\item\label{prop:ummf} If for all $i\in N: v_i(t)=\frac{d}{dt}\big(U_i^{-1}(t)\big)$ and $u_i(S)=U_i(V_i(t_i(S)))$, Algorithm~\ref{alg_pf} computes a utility max-min fair bandwidth allocation under $S$.
\end{enumerate}
\end{proposition}
\begin{proof}
As \eqref{prop:mmf} and \eqref{prop:wmmf} are known in the literature (cf.~\cite{BertsekasG92})
we only prove (\ref{prop:ummf}).
In order to obtain a utility max-min fair allocation, we need
to ensure that while raising rates, the bandwidth utilities must be equally distributed.
Thus, starting with $t=0$ we set  $U_i(V_i(t)=t$ for all $i\in N$.
This is equivalent to $U_i^{-1}(t)=V_i(t)$ using that $U_i$ is strictly increasing
and thus invertible. Differentiating both sides leads to $v_i(t)=\frac{d}{dt}\Big(U_i^{-1}(t)\Big)$
as claimed. Since $U_i$ is strictly increasing, its inverse is also strictly increasing (and also nonnegative), hence,
$v_i(t)$ satisfies all assumptions needed. Now it follows
by standard arguments (cf. \cite{HP08}) that the resulting allocation is utility max-min fair.
\end{proof}


\begin{algorithm}[t]
\caption{Progressive Filling(PF)}
\label{alg_pf}
\textbf{Parameters:} A \pfg\  $\mathcal{G}=\left(N,R,\left(c_i\right)_{r\in R},\left(\mathcal{S}\right)_{i\in N}, \left(v_i\right)_{i \in N}\right)$ \\
\textbf{Input:} A state $S=\left(S_1,\dots,S_n\right)\in\mathcal{S}$.\\
\textbf{Output:} The bandwidth $b_i$ for each player $i\in N$.
\begin{algorithmic}[1]
\State $b_i\gets0,$\textbf{ for all }$i\in N; N^\prime\gets N$
\State $N_r\gets\left\{i\in N\mid r\in S_i\right\}; c^\prime_r\gets c_r,$\textbf{ for all }$r\in R$
\While{$N^\prime\neq\emptyset$}
\State \begin{varwidth}[t]{\linewidth}$t^\star\gets\min\{t^\prime\mid\exists r\in R$\label{alg_pf_linemin}\par
\hspace{2.2cm}\text{ with }$\sum_{i\in N_r\cap N^\prime}\int_0^{t^\prime}v_i\left(t\right)dt=c^\prime_r$\par\hspace{2.2cm}\text{ and } $N_r\cap N^\prime\neq\emptyset\}$\end{varwidth}
\State \begin{varwidth}[t]{\linewidth}\textbf{choose }$r^\star$\textbf{ with }$\sum_{i\in N_{r^\star}\cap N^\prime}\int_0^{t^\star}v_i\left(t\right)dt=c^\prime_{r^\star}$\par\hspace{1.5cm}\textbf{ and }$N_{r^\star}\cap N^\prime\neq\emptyset$\end{varwidth}
\For{\textbf{each }$i\in N_{r^\star}\cap N^\prime$}
\State $b_i\gets\int_0^{t^\star}v_i\left(t\right)dt$\label{alg_pf_line7}
\State $N^\prime\gets N^\prime\setminus\left\{i\right\}$
\For{\textbf{each }$r\in S_i$}
\State $c^\prime_r\gets c^\prime_r-b_i$\label{alg_pf_line10}
\EndFor
\EndFor
\EndWhile
\State \textbf{return }$\left(b_1,\dots,b_n\right)$
\end{algorithmic}
\end{algorithm}
%


%

\section{Existence of Equilibria}
We first study game-theoretic properties of a \pfg. We show that SE exist and moreover every sequence of improving deviations of coalitions converges to a SE.

\begin{theorem}
\label{thm_pfg_se}
	Every \pfg\ has a SE and every sequence of improving deviations of coalitions converges to a SE.
\end{theorem}

\begin{proof}
Let $\mathcal{G}$ be a PFG and $S$ a state in this game. For a player $i$, recall that
we denote by $t_i\left(S\right)$ the finishing time, i.e., the point in time when his bandwidth is fixed by Algorithm~\ref{alg_pf} on $S$. Likewise, we denote by $\tilde{t}_r\left(S\right)$ the point in time when resource $r$ gets saturated. In the remainder of the proof we crucially exploit the monotone relationship between the obtained bandwidth and the finishing time of every player. By the monotonicity of  the $V_i$'s, if a player strictly improves his obtained bandwidth by using an alternative strategy, then the new finishing time must strictly increase.

For a state $S$, we define a lexicographical potential function $\phi:\mathcal{S}\rightarrow\mathbb{R}_+^n$ as 
the vector of finishing times sorted in non-decreasing order, i.e.,
$\phi\left(S\right)=\left(t_{i_1}\left(S\right),\dots,t_{i_n}\left(S\right)\right)$
with $\{i_1,\ldots,i_n\} = N$ and $t_{i_j}(S) \le t_{i_{j+1}}(S)$.

The next lemma shows that in a state $S$ an improving move of a coalition $C$ to a state $T$ implies that $\phi(S) \prec \phi(T)$ where $\prec$ denotes the lexicographic ordering of vectors. Thus, a $\prec$-maximal state must be a SE. This   implies the existence of the potential function and thereby the theorem.
\end{proof}
\begin{lemma}
\label{lem_dual_greedy_se2}
\label{lemma_cpf}
Let $C\subseteq N$ be a coalition which has an improving move from $S=\left(S_1,\dots,S_n\right)$ to $T=\left(T_1,\dots,T_n\right)$ where $S,T\in\mathcal{S}$. Then we have
\begin{compactenum}[\hspace{.3cm}(a)]
	\item $t_i\left(T\right)\geq t_i\left(S\right)$, for all $i\in N$ with $t_i\left(S\right)\leq\min_{j\in C}t_j\left(S\right)$, and
	\item $t_i\left(T\right)> \min_{j\in C}t_j\left(S\right)$, for all $i\in N$ with $t_i\left(S\right)>\min_{j\in C}t_j\left(S\right)$.
\end{compactenum}
\end{lemma}
\begin{proof}
For some player $i$, note that we have $\tilde{t}_r\left(S\right)>t^\star$ for all $r\in T_i$ if and only if $t_i\left(S\right)>t^\star$ for some $t^\star\in\mathbb{R}$. Hence, it suffices to show that
\begin{compactenum}[\hspace{.3cm}(a$^\prime$)]
 	\item $\tilde{t}_r\left(T\right)\geq \tilde{t}_r\left(S\right)$, for all $r\in R$ with $\tilde{t}_r\left(S\right)\leq\min_{j\in C}t_j\left(S\right)$, and
	\item $\tilde{t}_r\left(T\right)>\min_{j\in C}t_j\left(S\right)$, for all $r\in R$ with $\tilde{t}_r\left(S\right)>\min_{j\in C}t_j\left(S\right)$.
\end{compactenum}

For all $r\in T_i$ for some $i\in C$, the claim directly follows because we have $t_i\left(T\right)>t_i\left(S\right)\geq\min_{i\in C}t_i\left(S\right)$ and thus $\tilde{t}_r\left(T\right)>\min_{i\in C}t_i\left(S\right)$. So let $r\in R$ such that $r$ is not used in $T$ by any player from $C$.

In $S$, no resource which is used by a player from $C$ has been saturated before $\min_{i\in C}t_i\left(S\right)$. Consequently, the bandwidth allocated to a player $i$ is identical at time $\min_{i\in C}t_i\left(S\right)$ in $S$ and $T$ for all $i\in N \setminus C$. Since in $T$ resource $r$ is used by exactly the same players from $N \setminus C$ as in $S$ and by no player from $C$, the residual capacity of $r$ at a time $t\leq\min_{i\in C}t_i\left(S\right)$ is in $T$ at least as high as in $S$.
 
This last result immediately implies (a$^\prime$). For (b$^\prime$), let $\tilde{t}_r\left(S\right)>\min_{j\in C}t_j\left(S\right)$. This means that the residual capacity at time $\min_{i\in C}t_i\left(S\right)$ is above zero in $S$ and hence also in $T$. By the continuity of the indefinite integrals of the allocation rate functions, we obtain $\tilde{t}_r\left(T\right)>\min_{j\in C}t_j\left(S\right)$. 
\end{proof}

Note that the above result applies to PFGs in full generality, that is, only requiring that the functions $V$ are non-negative, non-decreasing, and tend to infinity for $t$ going to infinity. We now show that the assumptions underlying this result cannot be relaxed. Clearly, relaxing non-negativity or relaxing the unboundedness of $V$ makes not much sense. Negative aggregated rates have no physical meaning, and for a bounded $V$ there exists a game with large enough capacities for which Algorithm~\ref{alg_pf} does not terminate. More interestingly, suppose we have an allocation rate function for which  the aggregated bandwidth $V(t)$ is non-monotonic. Note that this extension still allows to use Algorithm~\ref{alg_pf} to calculate the allocation via progressive filling. We show that for any such function, Theorem~\ref{thm_pfg_se} does not hold anymore. This is even true if we restrict to two-player games 
with symmetric strategy spaces.

\begin{theorem}
\label{thm_ext}
	Let $v$ be such that $V:\mathbb{R}_+\rightarrow\mathbb{R}_+,t^\prime\mapsto\int_0^{t^\prime}v\left(t\right)dt$ 
	satisfies $V\geq 0$ and $\lim_{t\rightarrow\infty}V(t)=\infty$.
	If $V(t)$ is not monotone, there is a two-player PFG $\mathcal{G}_v$ 
	with symmetric strategy spaces that does not have a PNE and only uses $v$ and one constant function as allocation rate functions.
\end{theorem}

\begin{proof}
Let $v$ be an allocation rate function such that
the aggregated rate function $V$ is not monotone.
By the continuity and non-negativity of $V$, there
is  $t_1>0$ such that for every $\epsilon>0$, there is
$t_2=t_2(\epsilon)\in (t_1,t_1+\epsilon)$ with $V(t_1)>V(t_2)$
(see~\cite[Lemma 3.1]{HarksK12}).
Thus, we can choose
$t_2$ satisfying $t_2< t_1+\epsilon$
for any $\epsilon>0$ to be specified later.
Since $v$ is Riemann integrable and thus on the interval $\left[0,t_2\right]$ bounded, its indefinite integral $V$ has a Lipschitz constant $\rho>0$ on $\left[0,t_2\right]$. 

We now describe the game $\mathcal{G}_v$ with two players $\left\{1,2\right\}$. We set $R=\left\{r_1,r_2,r_3\right\}$ with
$c_{r_1}=c_{r_2}=(\rho+1) t_1+V(t_1)$ and $c_3=(\rho+1) t_2+V(t_2)$.
Furthermore, the sets of strategies are
$\mathcal{S}_1=\mathcal{S}_2=\left\{\left\{r_1,r_3\right\},\left\{r_2,r_3\right\}\right\}.$
As allocation rate functions, we use $v_1\equiv v$ and $v_2\equiv \rho+1 $.
We claim that, whenever both players share one of the resources $r_1$ or $r_2$,
the shared resource is saturated at time $t_1$ and player $2$ gets bandwidth $(\rho+1) t_1$ while player $1$ gets bandwidth $V(t_1)$.
To see this, we use the Lipschitz inequality
 $\frac{V(t)-V(t_1)}{t_1-t}< (\rho+1)$ for all $t\in [0,t_1)$
 implying  $(\rho+1) t+ V(t)< (\rho+1) t_1 + V(t_1)$ for all $t\in [0,t_1)$.
On the other hand, whenever player $2$ is alone on either $r_1$ or $r_2$,
resource $r_3$
is saturated at time $t_2$
using again $\frac{V(t)-V(t_2)}{t_2-t}< (\rho+1)$ for all $t\in [0,t_2)$.
By choosing $t_2<t_1+V(t_1)/(\rho+1)$ (hence $t_2=t_2(\epsilon)$ with  $\epsilon=V(t_1)/(\rho+1)$)
we get $(\rho+1)t_2< (\rho+1) t_1+V(t_1)$
and, thus, none of the resources $r_1$ or $r_2$ gets saturated
before $t_2$.
Consequently, player $2$ gets bandwidth $(\rho+1)  t_2>(\rho+1)  t_1$ while player $1$ gets bandwidth $V(t_2)<V(t_1)$. 
Hence, there is no PNE. 
\end{proof}

\section{Max-Min-Fair \PFG s}
\label{sec:MMFG}

A special case of \pfg s arises if all players raise their bandwidth uniformly,  i.e., $v_i(t) =1$ for all $i \in N$. This leads to allocations that are max-min fair. 
We call such a game max-min-fair \pfg\ or MMFG.
More formally, let $S\in\mathcal{S}$ be a state and $\mathcal{A}=\left\{a\mid a\text{ is an allocation in $S$}\right\}$, then the unique $\preceq$-maximal $a^\star$ in $\mathcal{A}$ is the max-min fair allocation.
In the following, we will study the computational complexity and efficiency of SE and PE in MMFGs.

\subsection{Computing Equilibria}
\label{sec:MMFG_Comp}
Similar to~\cite{HarksHKS10}, we use a \textit{dual greedy algorithm}~\cite{Schrijver03} to compute strong equilibria.
Our dual greedy algorithm is allowed to query a \textit{strategy packing oracle} that solves the strategy packing problem which is the following:
The input is given by a set $R$ of $m\in\mathbb{N}$ resources, $n$ sets of strategies $\mathcal{S}_i\in\mathcal{P}\left(R\right)$, for all $i\in\left\{1,\dots,n\right\}$,  
  along with upper bounds $u_r\in\left\{0,\dots,n\right\}$, for each $r\in R$.
The  output is a state $\left(S_1,\dots,S_n\right)\in\mathcal{S}_1\times\dots\times\mathcal{S}_n$ satisfying the upper bounds, i.e., $\left|\{i\in N\mid r\in S_i\}\right|\leq u_r$, for all $r\in R$, if it exists. Otherwise the output is the information that no such state exists.

 \begin{algorithm}[h]
\caption{Dual Greedy Algorithm}
\label{alg_dg}
Let $\mathfrak{O}$ denote the strategy packing oracle.\\
\textbf{Input:}  A MMFG  $\mathcal{G}=\left(N,R,\left(c_i\right)_{r\in R},\left(\mathcal{S}\right)_{i\in N}\right)$ \\
\textbf{Output:} A SE in $\mathcal{G}$.
\begin{algorithmic}[1]
\State $b_i\gets0,$\textbf{ for all }$i\in N; N^\prime\gets N$
\State $u_r\gets n, c^\prime_r\gets c_r,$\textbf{ for all }$r\in R$
\While{$N^\prime\neq\emptyset$}
\State $\left(S_i^\prime\right)_{i\in N^\prime}\gets\mathfrak{O}\left(R,\left(\mathcal{S}_i\right)_{i\in N^\prime},\left(u_r\right)_{r\in R}\right)$\label{alg2_oracle}
\State \textbf{choose }$r^\star\in\arg\min_{r\in E:u_r>0}\frac{c^\prime_r}{u_r}$\label{alg2_choice}
\State $u_{r^\star}\gets u_{r^\star}-1$\label{alg2_user_dec}
\If{$\mathfrak{O}\left(R,\left(\mathcal{S}_i\right)_{i\in N^\prime},\left(u_r\right)_{r\in R}\right)=\emptyset$}\label{alg2_oracle1}
\State $u_{r^\star}\gets u_{r^\star}+1$
\State $b\gets\frac{c^\prime_{r^\star}}{u_{r^\star}}$\label{alg2_bwcalc}
\For{\textbf{each }$i\in N^\prime$\textbf{ with }$r^\star\in S^\prime_i$}
\State $S_i\gets S^\prime_i$
\State $N^\prime\gets N^\prime\setminus\left\{i\right\}$
\For{\textbf{each }$r\in S_i$}\label{alg2_inner_for}
\State $u_r\gets u_r-1$
\State $c^\prime_r\gets c^\prime_r-b$
\EndFor\label{alg2_inner_for_end}
\EndFor
\EndIf
\EndWhile
\State \textbf{return }$S$
\end{algorithmic}
\end{algorithm}

The dual greedy algorithm initially allows an upper bound of $u_r=n$ players on each resource $r$ and every resource and every player is initially considered \emph{free}. The algorithm starts with an arbitrary state $S$ of strategies for players. It iteratively decrements one of the bounds $u_r$ on a free resource providing minimum bandwidth if each resource was used by $u_r$ players. After each decrement, it checks the existence of a strategy profile respecting the new upper bounds on the number of players using it by invoking the strategy packing oracle. When a decrease produces infeasible bounds, i.e., when there is no state of the game respecting the new bounds, it reverts the last decrease. Now we know that in the profile that was returned by the oracle, exactly $u_r$ players are using $r$ and it is infeasible to further reduce $u_r$. Thus, the algorithm turns $r$ into a \emph{fixed} resource, and also fixes the $u_r$ players as well as their strategies. In addition, it decreases every resource capacity 
by the amount given to the $u_r$ fixed players in their strategies. Then it continues with the remaining players, resources, and residual capacities. For a formal statement of the algorithm see Algorithm~\ref{alg_dg}.

\begin{theorem}
	\label{thm_dg_correctness}
	The dual greedy algorithm computes a SE.
\end{theorem}

\begin{proof}
The main idea of the proof is similar to~\cite{HarksHKS10}, i.e., the iterative assignment of Dual Greedy yields a lexicographically maximal vector of bandwidths. Consider on each resource the residual capacity not yet assigned to fixed players. We can assume that this residual capacity is offered in equal shares to the remaining free players. Thus, the share of each free player only depends on the number of free players using it. Hence, as long as no players are fixed, the game can be seen equivalently as a bottleneck congestion game. In addition, once a resource and players are fixed, then the bandwidth of a fixed player is smaller than the equal share of residual capacity on every free resource he uses. This allows to inductively show correctness of the algorithm.

More formally, fix a run of the dual greedy algorithm on the input instance given in the formal description and denote the output by $S=\left(S_1,\dots,S_n\right)$. Furthermore, by $b_i$, for a player $i\in N$, we denote his bandwidth calculated 
just before his strategy was fixed. We start off with proving the following useful lemma.

\begin{lemma}
\label{lem_dual_greedy_se}
Consider the $t$-th run of the main loop in Algorithm~\ref{alg_dg} where $t>1$. If $u_r>0$, the value of $\frac{c^\prime_r}{u_r}$ is not smaller than the value in the $\left(t-1\right)$-th run of the main loop.
\end{lemma}
\begin{proof}[Proof (Lemma)]
Observe that, in one run of the main loop, the fraction $\frac{c^\prime_r}{u_r}$ for some $r\in R$ can only be changed for the following two reasons.

\textit{Case 1:}
The resource $r$ is chosen in line~\ref{alg2_choice} and the oracle does not evaluate to $\emptyset$ in line~\ref{alg2_oracle1}. Then $u_r$ is decremented, i.e., the above fraction is increased.

\textit{Case 2:}
A resource $r^\star$ (not necessarily $r\neq r^\star$) is chosen in line~\ref{alg2_choice}, the oracle evaluates to $\emptyset$ in line~\ref{alg2_oracle1} and $r$ occurs in $k$ different strategies $S_i^\prime$ obtained from the oracle in line~\ref{alg2_oracle} where $1\leq k\leq u_{r^\star}\leq n$. According to the calculations from line~\ref{alg2_inner_for} to line~\ref{alg2_inner_for_end}, the new value of the above fraction is
$$\frac{c^\prime_r-k\cdot\frac{c^\prime_{r^\star}}{u_{r^\star}}}{u_r-k}=\frac{u_{r^\star}\cdot c^\prime_r-k\cdot c^\prime_{r^\star}}{u_{r^\star}\cdot u_r- k\cdot u_{r^\star}}$$
where we let $k<u_r$ since, otherwise, the new $u_r$ is $0$. Further, we have
$$\frac{u_{r^\star}\cdot c^\prime_r-k\cdot c^\prime_{r^\star}}{u_{r^\star}\cdot u_r- k\cdot u_{r^\star}}\geq \frac{c^\prime_r}{u_r}$$  which is equivalent to 
$$u_{r^\star}\cdot c^\prime_r-k\cdot c^\prime_{r^\star}\geq u_{r^\star}\cdot c^\prime_r-k\cdot u_{r^\star}\cdot\frac{c^\prime_r}{u_r}$$
and because of the choice of $r^\star$ gives us $c^\prime_{r^\star}/u_{r^\star} \leq c^\prime_r/u_r,$ which implies the claim. 
\end{proof}


%

Now let $N_k$ be the set of players whose strategies are fixed as a consequence of the oracle's $k$-th evaluation to $\emptyset$. We show by induction on $k$ that none of the players from $N_1\uplus\dots\uplus N_k$ will be part of a coalition performing an improving move, for all $k$. Note that this proves the theorem because we have $N=N_1\uplus\dots\uplus N_l$ for some $l\in\mathbb{N}$.

The base case of $k=0$ follows trivially. Now assume that the statement holds for some $k<l$. To see that this implies the statement for $k+1$, observe that the strategies $S_i$ and bandwidths $b_i$ of the players $i\in N_1\uplus\dots\uplus N_k$ are already fixed. Now suppose there is a coalition $C\subseteq N$ with $C\cap N_{k+1}\neq\emptyset$ profitably deviating from $S$ to $T=\left(S^\prime_C,S_{-C}\right)$. We consider the state of the variables at line~\ref{alg2_bwcalc} after the oracle's $k+1$-th evaluation to $\emptyset$.

Since $N_1,\dots,N_k$ are not participating in the improvement step, Lemma~\ref{lem_dual_greedy_se} implies that $N_{k+1}\in\argmin_{j\in C}b_j\left(S\right)$. Thus, for $i^\star\in N_{k+1}$, Lemma~\ref{lem_dual_greedy_se2} can be used to obtain that

\begin{itemize}
	\item $b_{i}\left(T\right)\geq b_{i^\star}\left(S\right)$, for all $i\in N$ with $b_{i}\left(S\right)=b_{i^\star}\left(S\right)$, and
	\item $b_{i}\left(T\right)> b_{i^\star}\left(S\right)$, for all $i\in N$ with $b_{i}\left(S\right)>b_{i^\star}\left(S\right)$.
\end{itemize}

 Again by Lemma~\ref{lem_dual_greedy_se}, this means that we have $l_r\left(T\right)\leq u_r$, for all $r\in R$. In particular, we even have $l_{r^\star}\left(T\right)< u_{r^\star}$ as $C\cap N_{k+1}\neq\emptyset$ and the players from $C$ strictly improve. Such a state $T$ may, however, not exist by the evaluation of the oracle to $\emptyset$. 
\end{proof}
 

Dual Greedy can be implemented in polynomial time given an efficient strategy packing oracle. Hence, the problem of computing SE in MMFGs is polynomial-time reducible to the strategy packing problem. There are several non-trivial cases in which the strategy packing problem is polynomial-time solvable, e.g., for single-commodity networks~\cite{HarksHKS10}. Thus, we obtain the following result.

\begin{corollary}
 \label{cor_mmfg_scnpoly}
 SE can be computed in polynomial time for single-commodity network MMFGs.
\end{corollary}

In contrast, the strategy packing problem turns out to be NP-hard even if we generalize to symmetric (non-network) strategy spaces. 

\begin{theorem}
	\label{SPack_hard}
  The strategy packing problem for symmetric strategies is NP-hard.
\end{theorem}

\begin{proof}
We reduce from the strongly NP-hard set packing problem.
Given an instance of the set packing problem  $\mathcal{I}=\left(\mathcal{U},\mathcal{S},k\right)$. From $\mathcal{I}$, we construct the following strategy packing instance $\mathcal{J}$. As resource set, we  choose $\mathcal{U}$ and, for each upper bound, we choose $u_r=1$. Furthermore, we set $\mathcal{S}_1=\dots=\mathcal{S}_k=\mathcal{S}$ for the strategy sets.
It is easy to see that there exists a set packing in $\mathcal{I}$ if and only if there exists a strategy packing in $\mathcal{J}$. This is because each family of subsets $\mathcal{S}^\prime\subseteq\mathcal{S}$ gives a state in $\mathcal{J}$  and vice versa. Obviously, $\mathcal{S}^\prime$ has mutually disjoint elements if and only if the corresponding state is satisfies  the upper bounds. 
\end{proof}

This result permits computation of SE polynomial time by other algorithms than Dual Greedy, but, in fact, computation of SE and strategy packing are \textit{mutually} polynomial-time reducible, even for symmetric games.

\begin{theorem}
	\label{SE_hard}
  The computation of a SE in symmetric MMFGs is NP-hard.
\end{theorem}

\begin{proof}
We reduce the strategy packing problem to the computation of SE  in symmetric MMFGs.
 Let $\mathcal{I}=\left(R,\left(\mathcal{S}_i\right)_{i\in\left\{1,\dots,n\right\}},\left(u_r\right)_{r\in R}\right)$ be an instance of the symmetric strategy packing problem, i.e., we have $\mathcal{S}_1=\dots=\mathcal{S}_k$. We create a symmetric MMFG $\mathcal{G_I}$ the following way. As resources, we define
\[R^\prime=R\cup\left\{r_1,\dots,r_n,r^\prime_1,\dots,r^\prime_n,r^\star\right\}.\]
The set of strategies for each of the $n+1$ players is defined by
\[\mathcal{S}^\prime_i=\left\{S_1\cup\left\{r_j,r_j^\prime\right\}\mid 1\leq j\leq n \wedge S_1\in\mathcal{S}_1\right\}\cup\left\{\left\{r_1,\dots,r_n,r^\star\right\},R\cup\left\{r^\star\right\}\right\}.\]
Finally, we set
\[c_r=\begin{cases}
        u_r+1, & \text{if }r\in R\\
	2-\varepsilon, & \text{if }r=r_i\text{ for some }i\\
	1, & \text{if }r=r_i^\prime\text{ for some }i\text{ or }r=r^\star
      \end{cases}
\]
as the capacity for each resource $r\in R$ where we choose $\varepsilon<\frac{1}{\left(n+1\right)^{n+1}}$. This defines a unique MMFG $\mathcal{G_I}$ with bandwidth functions $\left(b_i\right)_{i\in\left\{1,\dots,n+1\right\}}$. The model is illustrated in Figure~\ref{fig_smmfg_nphard}. Furthermore, this is obviously a polynomial time reduction (assuming $\varepsilon$ is chosen accordingly).

\begin{figure}
\centering
\subfiguretopcaptrue
\subfigure[][]{\begin{tikzpicture}[scale=1.1]
\definecolor{darkred}{rgb}{.85,0,0}

\node[draw,rectangle,minimum width=0.9cm,minimum height=0.9cm,rounded corners,very thick,fill=black!10] (0) at (3.5,5.2) {\tiny{\shortstack{$r^\star$\\$1$}}};
\foreach \y/\name in {1/n,3/2,4/1}{
\foreach \x in {0,1,...,2}
\node[draw,rectangle,minimum width=0.9cm,minimum height=0.9cm,rounded corners,very thick,fill=black!10] (\x\y) at (\x,\y) {\tiny{\shortstack{$r\in R$\\$u_r$+$1$}}};
\node[draw,rectangle,minimum width=0.9cm,minimum height=0.9cm,rounded corners,very thick,fill=black!10] (a\y) at (3.5,\y) {\tiny{\shortstack{$r_{\name}$\\$2$--$\varepsilon$}}};
\node[draw,rectangle,minimum width=0.9cm,minimum height=0.9cm,rounded corners,very thick,fill=black!10] (b\y) at (4.6,\y) {\tiny{\shortstack{$r_{\name}^\prime$\\$1$}}};
}
\foreach \x in {1,3.5,4.6}
\node (0) at (\x,2) {$\vdots$};

\node[draw=darkred,rectangle,very thick,minimum width=2.2cm,minimum height=1cm,rounded corners,fill=darkred,opacity=.3] (0) at (4.05,4) {};
\node[draw=darkred,rectangle,very thick,minimum width=2.2cm,minimum height=1cm,rounded corners,fill=darkred,opacity=.3] (0) at (4.05,3) {};
\node[draw=darkred,rectangle,very thick,minimum width=2.2cm,minimum height=1cm,rounded corners,fill=darkred,opacity=.3] (0) at (4.05,1) {};
\draw[draw=darkred,rounded corners,very thick,fill=darkred,opacity=.3] (0.45,2.5) -- (0.45,3.55) -- (2.45,3.55) -- (2.45,4.45) -- (-0.45,4.45) -- (-0.45,2.5);
\draw[draw=darkred,rounded corners,very thick,fill=darkred,opacity=.3] (1.55,2.55) -- (0.55,2.55) -- (0.55,4.45) -- (1.45,4.45) -- (1.45,3.45) -- (2.45,3.45) -- (2.45,2.55) -- (1.55,2.55);
\draw[draw=darkred,rounded corners,very thick,fill=darkred,opacity=.3] (-0.45,1.5) -- (-0.45,0.55) -- (2.45,0.55) -- (2.45,1.5);
\draw[draw=darkred,very thick,fill=darkred,opacity=.3] (2.45,4) -- (3.05,4);
\draw[draw=darkred,very thick,fill=darkred,opacity=.3] (2.45,4) -- (3.05,3);
\draw[draw=darkred,very thick,fill=darkred,opacity=.3] (2.45,4) -- (3.05,1);
\draw[draw=darkred,very thick,fill=darkred,opacity=.3] (2.45,3) -- (3.05,4);
\draw[draw=darkred,very thick,fill=darkred,opacity=.3] (2.45,3) -- (3.05,3);
\draw[draw=darkred,very thick,fill=darkred,opacity=.3] (2.45,3) -- (3.05,1);
\draw[draw=darkred,very thick,fill=darkred,opacity=.3] (2.45,1) -- (3.05,4);
\draw[draw=darkred,very thick,fill=darkred,opacity=.3] (2.45,1) -- (3.05,3);
\draw[draw=darkred,very thick,fill=darkred,opacity=.3] (2.45,1) -- (3.05,1);
\node (1) at (0.2,2) {\textcolor{darkred}{\footnotesize{$\in\mathcal{S}_1$}}};
\draw[->,very thick,draw=darkred] (1) -- (11);
\draw[->,very thick,draw=darkred] (1) -- (03);
\draw[->,very thick,draw=darkred] (1) -- (13);
\end{tikzpicture}}\hspace{0.5cm}
\subfigure[][]{\begin{tikzpicture}[scale=1.1]
\definecolor{darkblue}{rgb}{0,0,.85}

\node[draw,rectangle,minimum width=0.9cm,minimum height=0.9cm,rounded corners,very thick,fill=black!10] (0) at (3.5,5.2) {\tiny{\shortstack{$r^\star$\\$1$}}};
\foreach \y/\name in {1/n,3/2,4/1}{
\foreach \x in {0,1,...,2}
\node[draw,rectangle,minimum width=0.9cm,minimum height=0.9cm,rounded corners,very thick,fill=black!10] (0) at (\x,\y) {\tiny{\shortstack{$r\in R$\\$u_r$+$1$}}};
\node[draw,rectangle,minimum width=0.9cm,minimum height=0.9cm,rounded corners,very thick,fill=black!10] (0) at (3.5,\y) {\tiny{\shortstack{$r_{\name}$\\$2$--$\varepsilon$}}};
\node[draw,rectangle,minimum width=0.9cm,minimum height=0.9cm,rounded corners,very thick,fill=black!10] (0) at (4.6,\y) {\tiny{\shortstack{$r_{\name}^\prime$\\$1$}}};
}
\foreach \x in {1,3.5,4.6}
\node (0) at (\x,2) {$\vdots$};
\draw[draw=darkblue,rounded corners,very thick,fill=darkblue,opacity=.3] (-0.45,2.45) -- (-0.45,0.55) -- (2.45,0.55) -- (2.45,1.5) -- (2.45,4.45) -- (3.95,4.75) -- (3.95,5.65) -- (3.05,5.65) -- (1.45,4.45) -- (-0.45,4.45) -- (-0.45,2.45);
\draw[draw=darkblue,rounded corners,very thick,fill=darkblue,opacity=.3] (3.55,5.65) -- (3.95,5.65) -- (3.05,5.65) -- (3.05,0.55) -- (3.95,0.55) -- (3.95,5.65) -- (3.55,5.65);
\end{tikzpicture}}
\caption{Illustration of the strategies in the proof of Theorem~\ref{SE_hard}. The strategies built from the strategies in the strategy packing instance where a line between two sets indicates that the union is in $\mathcal{S}^\prime$ (a) as well as the two strategies independent of the strategies in strategy packing instance (b) are shown. Note that the resources in (a) and (b) are identical.}
\label{fig_smmfg_nphard}
\end{figure}
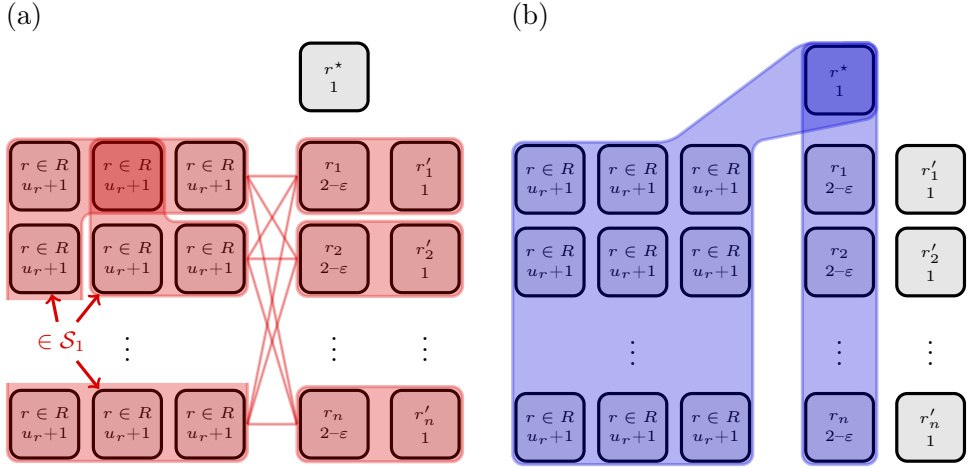

We will now show that, from a SE $S$ in $\mathcal{G_I}$ such that each player gets at least a bandwidth of $1$ in $S$, we can construct a strategy packing in $\mathcal{I}$. Conversely, we will show that the existence of any other SE already certifies that no such strategy packing exists.

The easier direction is the following. If in a state $S=\left(S_1,\dots,S_n\right)$ in $\mathcal{G_I}$, we have $b_i\left(S\right)\geq 1$, for all $i\in\left\{1,\dots,n+1\right\}$, at most for one player $i^\star$ we can have $r^\star\in S_{i^\star}$ because $c_{r^\star}=1$. Furthermore, as either $r^\star$ or some $r^\prime_i\in R^\prime$ occur in every strategy and the latter resources only allow a total bandwidth of $n$, there must exist \textit{at least} one such player. Thus, there exists a \textit{unique} player $i^\star$. 

Moreover, we must have $S_{i^\star}=R\cup\left\{r^\star\right\}$ since, otherwise, $r_1,\dots,r_n\in S_{i^\star}$ would hold, i.e., all the other players could get at most bandwidth $\frac{2-\varepsilon}{2}$. Consequently, at most $u_r$ players from $N\setminus\left\{i^\star\right\}$ may use a certain resource $r\in R$. As, for each $i\in N\setminus\left\{i^\star\right\}$, it must hold that $S_i\setminus\left\{r_1,\dots,r_n,r^\prime_1,\dots,r^\prime_n\right\}\in \mathcal{S}_1\subseteq R$, $\left(S_i\setminus\left\{r_1,\dots,r_n,r^\prime_1,\dots,r^\prime_n\right\}\right)_{i\in N\setminus\left\{i^\star\right\}}$ is therefore a strategy packing in $\mathcal{I}$.

For the other direction, we first introduce a lemma that, informally speaking, says that $\varepsilon$ is small enough.

\begin{lemma}
  Let $\mathcal{G}$ be a MMFG such that the capacities $u_r$ are integral, for each $r\in R$. Then, for each $S\in\mathcal{S}$, there is a $\delta\in\mathbb{N}$ with $\delta\leq n^n$ such that, for each player $i\in N$, the bandwidth $b_i\left(S\right)$ is $\frac{1}{\delta}$-integral.
  \label{lem_integral}
\end{lemma}
\begin{proof}[Proof (Lemma)]
Let $S\in\mathcal{S}$ and fix a run of Algorithm~\ref{alg_pf} on $S$. Define $N_1,\dots,N_k$ to be the partition of $N$ where $N_i$ is the set of players which are fixed in the $i$-th run of the main loop. By induction on $i$, we will now show that, for each $i\in\left\{1,\dots,k\right\}$, there is a $\delta_i\in\mathbb{Q}$ with $\delta_i\leq n^i$ such that all bandwidths of players in $\biguplus_{j\leq i}N_j$ as well as all the values of $c^\prime$ are $\frac{1}{\delta_i}$-integral at the end of the $i$-th run of the main loop. This already implies the claim since $k\leq n$.

For $i=0$, there is nothing to be shown. Now let $\delta_i$ as above. Since the values of $c^\prime$ are $\frac{1}{\delta_i}$-integral, the bandwidth calculated in line~\ref{alg_pf_linemin} and assigned in line~\ref{alg_pf_line7} in the $i+1$-th run of the main loop is $\frac{1}{\left|N_{i+1}\right|\cdot\delta_i}$-integral. The same holds for the values of $c^\prime$ changed in line~\ref{alg_pf_line10}. So we set $\delta_{i+1}:=\left|N_{i+1}\right|\cdot\delta_i$. By $\left|N_{i+1}\right|\leq n$, we have $\delta_i\leq n^{i+1}$. 
\end{proof}

Now let $S=\left(S_1,\dots,S_n\right)$ be a SE in $\mathcal{G_I}$ such that there exists a player $i$ with $b_i\left(S\right)<1$ and consider three different cases.

\begin{description}
 \item[Case 1:] There exists no player $j$ such that $r^\star\in S_j$. By Lemma~\ref{lem_integral}, we know that $b_i\left(S\right)<1-\varepsilon$. So it is profitable for this player to unilaterally deviate to the strategy $\left\{r_1,\dots,r_n,r^\star\right\}$ yielding a bandwidth of at least $1-\varepsilon$ for him. Hence, $S$ is no SE, in contradiction to our assumption.\label{proof_smmfg_nphard_case1}
 \item[Case 2:] There exists a player $j$ such that $S_j=\left\{r_1,\dots,r_n^\prime,r^\star\right\}$ but no player $k$ exists with $S_k=R\cup\left\{r^\star\right\}$. This means that \textit{all} players get a bandwidth of less than $1$ in $S$ (because either $r^\star$ or a resource $r_i^\prime$ occurs in each strategy). Therefore, they would all profitably deviate to a state $T$ with $b_i\left(T\right)=1$, for all players $i$. We now show that, however, such a state would exist in $\mathcal{G_I}$ if there was a strategy packing in $\mathcal{I}$. This immediately implies that there is no strategy packing in $\mathcal{I}$.
 
Let $\left(S^\prime_1,\dots,S^\prime_n\right)$ be the strategy packing in $\mathcal{I}$. In $T$, player $i$ uses the strategy $S^\prime\cup\left\{r_i,r_i^\prime\right\}$, for $i\in\left\{1,\dots,n\right\}$. Further, player $n+1$ uses the strategy $R\cup\left\{r^\star\right\}$. It can easily be verified that each player gets bandwidth $1$ in this state.\label{proof_smmfg_nphard_case2}
 \item[Case 3:] There exists a player $j$ such that $S_j=R\cup\left\{r^\star\right\}$. We again distinguish two cases.
\begin{description}
\item[Case a:] First, consider the case where another player $k$ exists with $k\neq j$ and $r^\star\in S_{k}$. This means that both players $j$ and $k$ get a bandwidth of less than $1$. Furthermore, each other player must also get a bandwidth of less than $1$ since, otherwise (i.e., if there exists a player $l$ getting at least bandwidth $1$), player $j$ could unilaterally and profitably deviate the following way. Player $j$ imitates player $l$ on $R$ and moreover chooses resources $\left\{r_{i^\star},r^\prime_{i^\star}\right\}$ such that $r^\prime_{i^\star}$ is not used in $S$, giving him a bandwidth at least as large as $b_j\left(S\right)$. With the same argumentation as is Case~\ref{proof_smmfg_nphard_case2}, we can hence infer that there is no strategy packing in $\mathcal{I}$.

\item[Case b:] Now let player $j$ be the unique player with $r^\star\in S_j$ and further let $S_j=R\cup\left\{r^\star\right\}$. We show again that $b_j\left(S\right)<1$  and apply the same argumentation as in Case~\ref{proof_smmfg_nphard_case1} (the preconditions of Lemma~\ref{lem_integral} are fulfilled since only resources with integral capacities are saturated). If $i=j$, we are finished. So let $i\neq j$ and suppose that, in $S$, each resource $r\in R$ is used by at most $u_r+1$ players. Since $S$ is a SE, we know that, in this case, each $r^\prime_i$ is used by at most one player. Hence, Algorithm~\ref{alg_pf} calculates a bandwidth of $1$ for each player; contradiction. Thus, there is a resource $r\in R$ used by more than $u_r+1$ players.

If all $r_i^\prime$ are used by one player each, there is a resource in $R$ that is the first one saturated in Algorithm~\ref{alg_pf} (by the existence of $r$), which implies the claim. So let $i^\star\in\left\{1,\dots,n\right\}$ such that $r^\prime_{i^\star}$ is a resource not used in $S$ and suppose $b_i\left(S\right)<b_j\left(S\right)$. Then, player $i$ could, however, replace the resources from $\left\{r_1,\dots,r_n,r_1^\prime,\dots,r_n^\prime\right\}$ he currently uses by $\left\{r_{i^\star},r_{i^\star}^\prime\right\}$, resulting in a bandwidth at least as large as $b_j\left(S\right)$.
\end{description}
\end{description} 
\end{proof}

\subsection{Efficiency of Equilibria}
\label{sec:MMFG_PoA}

In this section we investigate the quality of SE in terms of social welfare, i.e., the sum of allocated bandwidth. In a game $\mathcal{G}$, let $S^\star$ with allocation $a$ be the state in $\mathcal{S}$ that maximizes $\sum_{i \in N} a_i$. Further, let $\mathcal{S}^{SE}\subseteq\mathcal{S}^{NE}\subseteq\mathcal{S}$ denote the set of SE and NE, respectively. We denote $\SW_{\mathcal{G}}(S) = \sum_{i \in N} b_i(S)$. Then, the price of stability and price of anarchy, PoS and PoA, are defined as 
$\inf_{S\in \mathcal{S}^{NE}} \frac{\SW_{\mathcal{G}}(S^\star)}{\SW_{\mathcal{G}}(S)} = \inf_{S\in \mathcal{S}^{NE}} \frac{\sum_{i \in N} a_i}{\sum_{i \in N} b_i(S)}$ and $\sup_{S\in \mathcal{S}^{NE}} \frac{\SW_{\mathcal{G}}(S^\star)}{\SW_{\mathcal{G}}(S)}$, respectively. For the strong price of stability and anarchy, SPoS and SPoA, $\mathcal{S}^{SE}$ is considered instead of $\mathcal{S}^{NE}$. Furthermore, the same measures can be applied to classes of games where they are simply the supremum of all individual measures.

The {\em maximum capacity allocation problem (MCAP)} is given by the problem of computing an allocation $a^\prime$ which maximizes $\sum_{i \in N} a^\prime_i$. Note that we have $\sum_{i \in N} a^\prime_i\geq\sum_{i \in N} a_i$ and that this inequality may even be strict since $a^\prime$ is not necessarily computed by progressive filling.

In general, one cannot hope to find SE with good social welfare. There are network MMFGs in which even the best PNE is a factor of $\Omega\left(n\right)$ worse than the optimum. This matches the upper bound of $O(n)$ on the PoA for network MMFGs shown in~\cite{YangXFMZ10}. 

\begin{theorem}
  \label{thm_pos_nmmfgs}
  The PoS and SPoS in multi-commodity network MMFGs are $\Omega\left(n\right)$.
\end{theorem}

\begin{proof}
 For a given $n\in\mathbb{N}$, we construct a network MMFG $\mathcal{G}_{n}$ with $n$ players and PoS of more than $\frac{n}{4}$. 
 We assume w.l.o.g. that $2\mid n$.
 
 The network underlying $\mathcal{G}_{n}$ consists of $\frac{n}{2}$ consecutive edges each of which has capacity $1$ and connects the source and sink nodes $s_i, t_i$ of one respective player $i$. The source and sink nodes of the other $\frac{n}{2}$ players are the first and last vertex of this path. Additionally, there is one edge with the capacity $0$ between these two vertices. This network is illustrated in Figure~\ref{fig_pos_nmmfgs}.

At first, note that there are two strategies for each of the players. A player can either choose the path through the $1$-edges or the path which has the $0$-edge in it. The latter path will, however, not be taken in a PNE since avoiding the $0$-edge always results in a bandwidth strictly larger than $0$. Thus, in the unique PNE $S$, each $1$-edge is congested with $\frac{n}{2}+1$ players, resulting in a social welfare of
$\SW_{\mathcal{G}_n}\left(S\right)=n\cdot\frac{1}{\frac{n}{2}+1}=\frac{2n}{n+2}.$

If, however, the players $i\in\left\{\frac{n}{2}+1,\dots,n\right\}$ altruistically take the direct $0$-capacity path $\left(s_i,t_i\right)$ instead, all the other players get a bandwidth of $1$ by sticking to their paths from $S$. Consequently, a lower bound on the PoS is $\frac{\frac{n}{2}}{\SW_{\mathcal{G}_n}\left(S\right)}=\frac{\frac{n}{2}}{\frac{2n}{n+2}}=\frac{n+2}{4}>\frac{n}{4}.$ 
\end{proof}

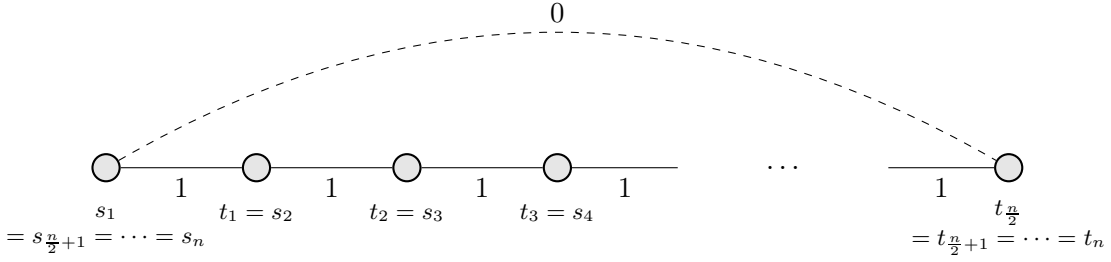
\begin{figure}
\centering
\begin{tikzpicture}[scale=2]
\foreach \x/\name in {0/$s_1$,1/$t_1=s_2$,2/$t_2=s_3$,3/$t_3=s_4$,6/$t_{\frac{n}{2}}$}{
\node[draw,circle,minimum width=0.08cm,thick,fill=black!10] (\x) at (\x,0) {};
\node[below of=\x,yshift=.4cm] (a\x) {\footnotesize{\name}};
}
\node[below of=a0,yshift=.6cm] {\footnotesize{$=s_{\frac{n}{2}+1}=\dots=s_n$}};
\node[below of=a6,yshift=.6cm] {\footnotesize{$=t_{\frac{n}{2}+1}=\dots=t_n$}};
\node (8) at (4.5,0) {$\hdots$};
\foreach \x/\y in {0/1,1/2,2/3}
\draw[-] (\x) --  node[below] {1} (\y);
\draw[-] (3) --  node[below] {1} (3.8,0);
\draw[-] (5.2,0) --  node[below] {1} (6);
\draw[-,dashed] (0) to  [bend left] node[above] {0} (6);
\end{tikzpicture}
\caption{Illustration of the network of the game $\mathcal{G}_{n}$ in the proof of Theorem~\ref{thm_pos_nmmfgs}.}
\label{fig_pos_nmmfgs}
\end{figure}

In contrast, when all players have the same strategy set, the best SE achieves a good approximation, and such a good SE is found by Dual Greedy (for single-commodity networks even in polynomial time).

\begin{theorem}
	\label{thm_dg_approx}
	The PoS and SPoS in symmetric MMFGs are $2-\frac{1}{n}$, and this bound is tight. The Dual Greedy computes an SE 
	achieving this guarantee. 
\end{theorem}

\begin{proof}
For the upper bound, we use an idea from~\cite{BaierKS05} and define the {\em uniform MCAP} as
the restriction of the MCAP to uniform bandwidth values, i.e., we additionally require that the found allocation is a vector $\left(a,\dots,a\right)$ for some $a\in\mathbb{R}$. It is easy to see that the smallest bandwidth in the state $S_{DG}$ computed by Dual Greedy solves the uniform MCAP. That is $\min_{i\in N}b_i\left(S_{DG}\right)=v$ where $n\cdot v$ is the optimal value of the uniform MCAP.

\begin{lemma}
	\label{lem_dg_umcap}
	Let $S=\left(S_1,\dots,S_n\right)$ be a solution of Dual Greedy on $\mathcal{G}$ and $n\cdot v$ be the optimal 
	value of the uniform MCAP, for $n\in\mathbb{N}$ and $v\in\mathbb{R}$. Then, we have $v=\min_{i\in N}b_i\left(S	
	\right).$
\end{lemma}
\begin{proof}[Proof (Lemma)]
We show the lemma in two steps:

$v\leq\min_{i\in N}b_i\left(S\right)$: By Lemma~\ref{lem_dual_greedy_se}, $\min_{i\in N}b_i\left(S\right)$ is exactly the bandwidth allocated to a player after the oracle has evaluated to $\emptyset$ for the first time. Again by Lemma~\ref{lem_dual_greedy_se}, such an evaluation to $\emptyset$ means that there is no state $S^\prime\in\mathcal{S}$ with $b_j\left(S^\prime\right)>\min_{i\in N}b_i\left(S\right)$, for all $i\in N$. This implies the claim since, otherwise, such a state $S^\prime$ is given by the optimal solution of the uniform MCAP.

$v\geq\min_{i\in N}b_i\left(S\right)$: Suppose $v<\min_{i\in N}b_i\left(S\right)$. Then we construct a feasible solution of the uniform MCAP with a larger value. We choose $S$ as state and $\min_{i\in N}b_i\left(S\right)$ as bandwidth for each player (which is feasible since we only possibly lower the feasible bandwidth $b_i\left(S\right)$, for all $i$). This may, however, not happen as the solution to the uniform MCAP has the value $n\cdot v$ by assumption.
\end{proof}

Thus, the upper bound follows from the next lemma.

\begin{lemma}
An optimum to the uniform MCAP is a $\left(2-\frac{1}{n}\right)$-approximation for the MCAP.
\label{lem_uniform_approx}
\end{lemma}

\begin{proof}[Proof (Lemma)]
 Let $n\cdot v$ be the optimal value of the uniform MCAP. Consider an arbitrary feasible solution attained by the state $S=\left(S_1,\dots,S_n\right)$ and the respective allocation $\left(a_1,\dots,a_n\right)$. Define $\alpha_i\in\mathbb{R}_+$ such that $a_i=\alpha_i\cdot v$, for all $i\in N$. 

  Now suppose that $\sum_{i\in N}\alpha_i>2n-1$. Then we can construct a new state $S^\prime=\left(S^\prime_1,\dots,S^\prime_n\right)$ with a corresponding bandwidth $v^\prime>v$ for each player. In $S^\prime$, we use at most $\lceil\alpha_i\rceil-1$ copies of the strategy $S_i$, for all $i$, and no other strategy. Furthermore, we set $v^\prime:=\min_{i\in N}{\frac{\alpha_i}{\lceil\alpha_i\rceil-1}}\cdot v$. Three properties of this solution remain to be shown:

\begin{compactenum}

 \item{There are at least $n$ (not necessarily different) strategies constructed for $S^\prime$ above. Using that $\lceil\alpha_i\rceil-1\geq\alpha_i-1$ holds for all $i$, we get:
$\sum_{i\in N}\left(\lceil\alpha_i\rceil-1\right)\geq\sum_{i\in N}\left(\alpha_i-1\right)=\left(\sum_{i\in N}\alpha_i\right)-n>n-1.$
As $\sum_{i\in N}\left(\lceil\alpha_i\rceil-1\right)$ must be integer, it follows that $\sum_{i\in N}\left(\lceil\alpha_i\rceil-1\right)\geq n$.}
 \item{The constructed bandwidths are feasible. We only use the strategies from $S$ for which $\left(a_1,\dots,a_n\right)$ is an allocation. So it suffices to see that for all $i$, in the constructed solution, the total capacity on $S_i$ is at most as high as $a_i$: 
$
\left(\lceil\alpha_i\rceil-1\right)\cdot v^\prime=\left(\lceil\alpha_i\rceil-1\right)\cdot\min_{j\in N}{\frac{\alpha_j}{\lceil\alpha_j\rceil-1}}\cdot v\leq
\left(\lceil\alpha_i\rceil-1\right)\cdot\frac{\alpha_i}{\lceil\alpha_i\rceil-1}\cdot v=a_i.
$

 \item{It indeed holds that $v^\prime>v$: by $\lceil\alpha_i\rceil-1<\alpha_i$, we have $\frac{\alpha_i}{\lceil\alpha_i\rceil-1}>1$ for all $i$, i.e., we also have $\min_{i\in N}{\frac{\alpha_i}{\lceil\alpha_i\rceil-1}}>1$ and hence $v^\prime>v$.}

}
\end{compactenum}
Thus, we have constructed a new solution of the uniform MCAP with a higher value than $n\cdot v$. So the initial solution cannot be maximal, i.e., we obtain a contradiction. So we must have 
$
\sum_{i\in N}\alpha_i\leq2n-1,
$
which implies 
$$\frac{\sum_{i\in N}a_{i}}{n\cdot v} \quad = \quad \frac{\sum_{i\in N}\alpha_{i}\cdot v}{n\cdot v} \quad = \quad \frac{\sum_{i\in N}\alpha_i}{n} \quad \leq \quad \frac{2n-1}{n}\enspace.$$ 
\end{proof}

For the lower bound consider for given $n\in\mathbb{N}$ and $\varepsilon\in\mathbb{R}$ with $\varepsilon>0$ a single-commodity network MMFG $\mathcal{G}_{n,\varepsilon}$.
From the source to the sink node, there are $n-1$ parallel edges each of which has capacity $1-\varepsilon$. Moreover, there is one single edge with capacity $n$.
  In the optimal state, every edge is used by one player each, i.e., we obtain a social welfare of $2n-1-\left(n-1\right)\cdot\varepsilon$. In a NE, however, every player uses the edge with capacity $n$ because the bandwidth for each player is at least $1>1-\varepsilon$ on this edge. Thus, the social welfare is exactly $n$ in this state. Therefore, it holds that
$$\PoS\left(\mathcal{G}_{n,\varepsilon}\right) \quad = \quad \frac{2n-1-\left(n-1\right)\cdot\varepsilon}{n} \quad = \quad 2-\frac{1}{n}-\frac{n-1}{n}\cdot\varepsilon\enspace.$$
\end{proof}

In symmetric games even the worst SE is still a $4$-approximation. For $n=2$, we can tighten the bound on the SPoA to the lower bound of the SPoS of $\frac{3}{2}$. 

\begin{theorem}
 The SPoA for symmetric MMFGs is at most $4 - \frac{6}{n+1}$.
 \label{spoa_4}
\end{theorem}

\begin{proof}
 Let $\mathcal{G}$ be a symmetric MMFG and let $S$ be a SE in this game. Then in $S$ each player must get at least a bandwidth of $\frac{1}{2}\max_{i\in N}b_i\left(S\right)$, as otherwise this player could profitably imitate a player in $\argmax_{i\in N}b_i\left(S\right)$ by choosing the same strategy. Thus, we can lower bound the social welfare by
 
\begin{equation}
\SW_\mathcal{G}\left(S\right) \quad = \quad \sum_{i\in N}b_i\left(S\right) \quad \geq \quad \left(\frac{n-1}{2}+1\right)\cdot\max_{i\in N}b_i\left(S\right).
\label{thm_spoa_eq1}
\end{equation}
 
State $S_{DG}$ computed by Dual Greedy in $\mathcal{G}$ is such that $\min_{i\in N}b_i\left(S_{DG}\right)=v$ where $n\cdot v$ is the optimal value of the uniform MCAP. Consequently, for any other SE $S$, we must have $\max_{i\in N}b_i\left(S\right)\geq v$, because otherwise all the players could profitably switch to their strategies in $S_{DG}$. Using Equation~\ref{thm_spoa_eq1}, this means $ \SW_\mathcal{G}\left(S\right)\geq\frac{n+1}{2}\cdot v$, and hence we obtain
 \begin{align*} 
 \frac{\max_{S^\prime\in\mathcal{S}}\SW_\mathcal{G}\left(S^\prime\right)}{\SW_\mathcal{G}\left(S\right)}
 &\leq \frac{2n}{n+1}\cdot\frac{\max_{S^\prime\in\mathcal{S}}\SW_\mathcal{G}\left(S^\prime\right)}{n\cdot v}\\
 & \leq \frac{2n}{n+1}\cdot\frac{2n-1}{n}=\frac{4n-2}{n+1}\enspace.
 \end{align*} 
\end{proof}

\begin{theorem}
\label{spoa2_32}
 The SPoA for symmetric MMFGs with $2$ players is $\frac{3}{2}$ and this bound is tight.
\end{theorem}

\begin{proof}
 Let $\mathcal{G}$ be a symmetric MMFG with $n=2$ and let $S$ be a SE in this game. Further, let $S^\prime$ be an arbitrary (optimal) state. W.l.o.g., we may assume that $b_1\left(S\right)\leq b_2\left(S\right)$ and $b_1\left(S^\prime\right)\leq b_2\left(S^\prime\right)$.
 
 Note that $b_1\left(S\right)\geq b_1\left(S^\prime\right)$ or $b_2\left(S\right)\geq b_2\left(S^\prime\right)$ must hold. Otherwise switching from $S$ to their strategies in $S^\prime$ would be profitable for both players. Thus, the following case distinction is complete.
 
\textit{Case 1:}
We have $b_1\left(S\right)\geq b_1\left(S^\prime\right)$. We can also derive an upper bound on $b_2\left(S^\prime\right)$. If $b_2\left(S^\prime\right)>2\cdot b_1\left(S\right)$, player $1$ could profitably deviate to $S^\prime_2$ in $S$. So we must have $b_2\left(S^\prime\right)\leq 2\cdot b_1\left(S\right)$. Thus,	
\[\frac{\SW_\mathcal{G}\left(S^\prime\right)}{\SW_\mathcal{G}\left(S\right)}=\frac{b_1\left(S^\prime\right)+b_2\left(S^\prime\right)}{b_1\left(S\right)+b_2\left(S\right)}\leq\frac{3\cdot b_1\left(S^\prime\right)}{2\cdot b_1\left(S\right)}\leq\frac{3}{2}\enspace.\]
\textit{Case 2:}
We have $b_2\left(S\right)\geq b_2\left(S^\prime\right)$. We find an upper bound on $b_2\left(S\right)$. Since $b_1\left(S\right)<\frac{1}{2}\cdot b_2\left(S\right)$ would mean that player $1$ could profitably imitate player $2$ in $S$, it holds that $b_1\left(S\right)\geq\frac{1}{2}\cdot b_2\left(S\right)$. This implies
\[\frac{\SW_\mathcal{G}\left(S^\prime\right)}{\SW_\mathcal{G}\left(S\right)}=\frac{b_1\left(S^\prime\right)+b_2\left(S^\prime\right)}{b_1\left(S\right)+b_2\left(S\right)}\leq\frac{2\cdot b_2\left(S^\prime\right)}{\frac{3}{2}\cdot b_2\left(S\right)}\leq\frac{4}{3}<\frac{3}{2}\enspace.\]
The lower bound immediately follows from Theorem~\ref{thm_dg_approx}. 
\end{proof}

In addition, we show a lower bound of $\Omega(n/k)$ on the $k$-SPoA for $k$-SE, where only deviations of coalitions of size at most $k$ are considered.

\begin{theorem}
 \label{thm_spoakn}
 The $k$-SPoA for single-commodity network MMFGs is in $\Omega\left(\frac{n}{k}\right)$.
\end{theorem}

\begin{proof}
We construct a family of single-commodity networks MMFG $\mathcal{G}_{n,k}$ with SPoA $\frac{n}{k}$. As we are showing an asymptotical lower bound, we may assume w.l.o.g. that $k\mid n$.
 
 The game $\mathcal{G}_{n,k}$ consists of $k$ gadgets $G_{n,i}$ for $i\in\left\{1,\dots,k\right\}$ where gadget $G_{n,i}=\left(V_{n,i},E_{n,i},c_{n,i}\right)$ is the following network. For the vertices and edges, we set
\begin{align*}
V_{n,i}=&\left\{u_i,v_{i,1},\dots,v_{i,n},w_{i,1},\dots,w_{i,n},u_{i+1}\right\},\\
E_{n,i}=&\left\{u_i\right\}\times\left\{v_{i,1},\dots,v_{i,n}\right\}\cup\left\{\left(v_{i,j},w_{i,j}\right)\mid 1\leq j\leq n\right\}\\
&\cup\left\{\left(w_{i,j},v_{i,j+1}\right)\mid 1\leq j\leq n-1\right\}\cup\left\{w_{i,1},\dots,w_{i,n}\right\}\times\left\{u_{i+1}\right\}
\end{align*}
and, further, we let $c_{n,i}\left(e\right)=1$, for all $e\in E$.

By arranging the $G_{i,n}$ in a row, we obtain the network underlying $\mathcal{G}_{n,k}$. More specifically, this network is 

\[\left(V_{n,1}\cup\dots\cup V_{n,k},E_{n,1}\cup\dots\cup E_{n,k}\cup E^\star_{n,k},c_{n,1}\cup\dots\cup c_{n,k}\cup c^\star_{n,k}\right)\]

with the source and sink nodes $u_1$ and $u_{k+1}$, respectively, and where

\[E^\star_{n,k}=\left\{\left(u_1,v_{i,1}\right)\mid 2\leq i\leq n\right\}\cup\left\{\left(w_{n,1},u_{n+1}\right)\mid 1\leq i\leq n-1\right\}.\]

Moreover, $c^\star_{n,k}$ is again constantly $1$ on $E^\star_{n,k}$. This network is illustrated in Figure~\ref{fig_kspoa}.

\begin{figure}
\centering
\begin{tikzpicture}[scale=1]
\definecolor{darkred}{rgb}{.85,0,0}
\definecolor{darkblue}{rgb}{0,0,.85}
\foreach \x/\namea in {0/a,3.6/b,10/d}{
\node[draw,circle,minimum width=0.08cm,thick,fill=black!10] (\namea) at (\x,0) {};
\foreach \y/\nameb in {-1.5/a,0.3/b,0.9/c,1.5/d}{
\node[draw,circle,minimum width=0.08cm,thick,fill=black!10,xshift=1.4cm] (\namea\nameb1) at (\x,\y) {};
\node[draw,circle,minimum width=0.08cm,thick,fill=black!10,xshift=2.2cm] (\namea\nameb2) at (\x,\y) {};
\draw[->,thick] (\namea) -- (\namea\nameb1);
\draw[->,line width=0.075cm] (\namea\nameb1) -- (\namea\nameb2);
}
\node[xshift=1.8cm] (0) at (\x,-0.6) {$\vdots$};
}
\node[draw,circle,minimum width=0.08cm,thick,fill=black!10] (c) at (7.2,0) {};
\node[draw,circle,minimum width=0.08cm,thick,fill=black!10] (e) at (13.6,0) {};
\node (0) at (8.6,0) {$\hdots$};
\foreach \nameb in {a,b,c,d}
\foreach \namea/\namec in {a/b,b/c,d/e}
\draw[->,thick] (\namea\nameb2) -- (\namec);

\foreach \namea in {a,b,d}{
\foreach \nameb/\namec in {d/c,c/b}
\draw[->,line width=0.075cm] (\namea\nameb2) -- (\namea\namec1);
\draw[->,line width=0.075cm] (\namea b2) -- ++(-0.5,-0.375);
\draw[<-,line width=0.075cm] (\namea a1) -- ++(0.5,0.375);
}

\draw[->,line width=0.075cm] (a) .. controls (1,2.3) .. (bd1);
\draw[->,line width=0.075cm] (a) .. controls (0.5,3) .. (dd1);
\draw[->,line width=0.075cm] (aa2) .. controls (13.1,-3) .. (e);
\draw[->,line width=0.075cm] (ba2) .. controls (12.7,-2.5) .. (e);
\draw[->,line width=0.075cm] (a) -- (ad1);
\draw[->,line width=0.075cm] (da2) -- (e);

\draw [decoration={brace,amplitude=0.5em},decorate,ultra thick,darkred] (14,1.7) -- node[right,xshift=.2cm] {\shortstack{$n$ disjoint\\paths per\\gadget}} (14,-1.7);
\draw [decoration={brace,amplitude=0.5em},decorate,ultra thick,darkred] (13.8,-2.7) -- node[below,yshift=-.3cm] {$k$ gadgets} (0,-2.7);

\node (0) at (2,3.5) {\textcolor{darkred}{$\frac{n}{k}$ players each}};
\draw[->,draw=darkred,very thick] (0) to [bend left] (1.9,2.45);
\draw[->,draw=darkred,very thick, bend angle=10] (0) to [bend left] (1,1.73);
\draw[->,draw=darkred,very thick] (0) to [bend right] (0.8,0.87);
\end{tikzpicture}
\caption{Illustration of the network of the game $\mathcal{G}_{n,k}$ with the SE $S^\star$ in the proof of Theorem~\ref{thm_spoakn}.}
\label{fig_kspoa}
\end{figure}
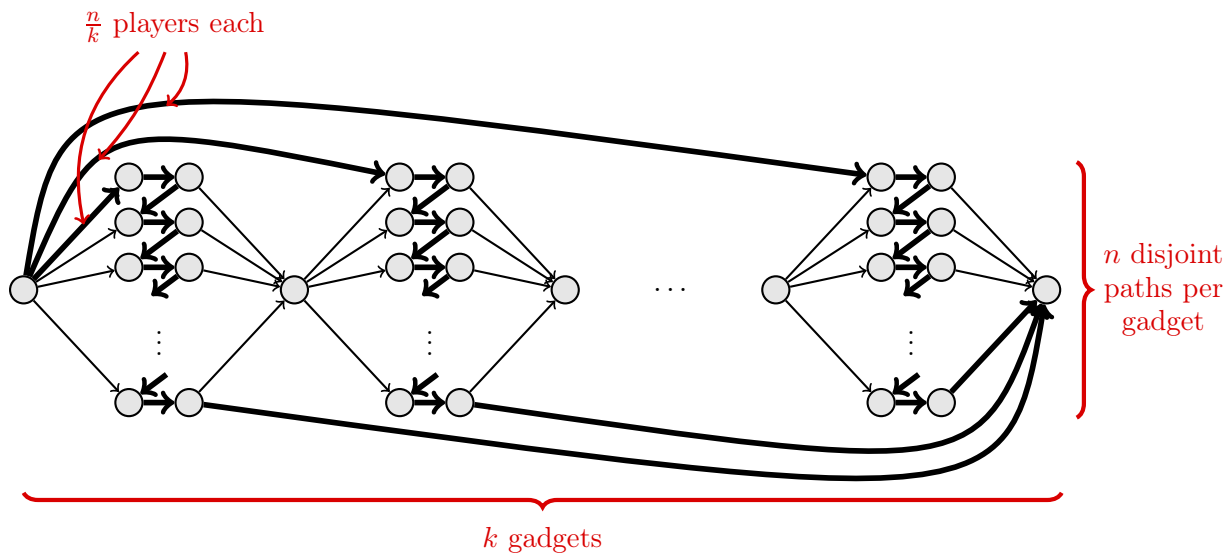

Since all the edges have capacity $1$, the optimal social welfare is $n$. In particular, in an optimal state player $i$ chooses the path $\left(u_1,v_{1,i},w_{1,i},u_{2},v_{2,i},\dots,w_{n,i},u_{n+1}\right)$, which can easily be verified. Thus, the optimal social welfare is $n$.

We now describe a state $S^\star$ (also shown in Figure~\ref{fig_kspoa}) that will be shown to be a $k$-SE and attain a social welfare of $k$, implying the claim. We partition the player set $N$ into $N_1,\dots,N_k$ with each player $j\in N_i$ choosing the path $\left(u_1,v_{i,1},w_{i,1},v_{i,2},w_{i,2},\dots,v_{i,n},w_{i,n},u_{i+1}\right)$, for all $i\in\left\{1,\dots,k\right\}$. Then, obviously, the paths of the players from different sets $N_i$ are pairwise arc-disjoint. Hence, the social welfare in $S^\star$ is indeed $k$.

It remains to be shown that $S^\star$ is indeed a $k$-SE. Towards a contradiction, suppose that there is a coalition $C\subseteq N$ with $\left|C\right|\leq k$ and a state $S^\prime$ such that each player in $C$ strictly improves when moving from $S^\star$ to $T=\left(S^\prime_C,S^\star_{-C}\right)$. We distinguish two cases and will use implicitly that in $S^\star$ there is no edge with more than $\frac{n}{k}$ players on it.

\begin{enumerate}[\hspace{.1cm}{Case} 1:]
 \item There is an $i\in\left\{1,\dots,k\right\}$ such that $N_i\cap C=\emptyset$. W.l.o.g., $i$ is the minimum $i$ with this property. Furthermore, we let $p$ be the number of such sets, i.e. $p=\left|\left\{j\mid N_j\cap C=\emptyset\right\}\right|$.

 First note that, if a player from $C$ passes through a gadget $G_{n,j}$ with $N_j\cap C=\emptyset$, the global bottleneck edge will be used by at least $k+1$ players and thus by at least one player from $C$. So each player from $C$ must use at least one edge from $E_C=\left\{\left(w_{j,n},u_{k+1}\right)\mid 1\leq j<i\right\}$ $\cup\left\{\left(u_1,v_{j,1}\right)\mid i\leq j\wedge N_j\cap C\neq\emptyset\right\}$. 

Now note that there are $n-p\cdot\frac{k}{n}-\left|C\right|$ players from $N \setminus C$ on $E_C$ and we have $\left|E_C\right|=k-p$. Thus, in $T$, there is a global bottleneck edge that has at least
\[\frac{n-p\cdot\frac{n}{k}}{k-p}=\frac{n}{k}\]
 players and, among them, one player from $C$ on it. This player does not strictly improve.
 \item For each $i\in\left\{1,\dots,k\right\}$, it holds that $\left|N_i\cap C\right|=1$. Then, on each path from $u_1$ to $u_{k+1}$ there is an edge used by exactly $\frac{n}{k}-1$ players from $N \setminus C$. Thus, adding the players from $C$ to them produces a global bottleneck edge with at least $\frac{n}{k}$ players on it. Consequently, the players from $C$ on the global bottleneck edge cannot strictly improve.   
\end{enumerate} 
\end{proof}

\section{General \PFG s}
\label{sec:GeneralPFG}

\subsection{Complexity and Convergence}
\label{sec:impSeq}

The lexicographical potential function for PFGs implies that the length of each coalitional improvement sequence is finite. By $\Phi_{\ord}$, we denote the set of ordered values of $\phi$, i.e., $\Phi_{\ord}:={\img}\left({\ord}\circ\phi\right)$ where $\ord$ is the function which orders a vector, say ascendingly. The cardinality of the above set provides an upper bound on the length of improvement sequences.
 
%
%
%

For a MMFG with $n$ players and $m$ resources, Yang et al.~\cite{YangXFMZ10} provide an upper bound of $\left(mn\right)^n$ on the number of improvement steps to reach a PNE. In the following, we show that it is not possible to get this result by just bounding $\left|\Phi_{\ord}\right|$.

\begin{theorem}
	\label{thm_china_falsch}
	There is a family of network MMFGs $\mathcal{G}_n$ with $m \in \Theta\left(n\right)$ and respective potential 
	function $\Phi_{\ord}$ such that $\left|\Phi_{\ord}\right|$ is in $2^{\Omega\left(n^2\right)} = \omega\left((n^2)^n
	\right)$.
\end{theorem}

\begin{proof}
Since we are proving an asymptotical lower bound, we may assume w.l.o.g. that $2\mid n$. We now describe the multigraph underlying $\mathcal{G}_n$. 

For each player $i\in\left\{\frac{n}{2}+1,\dots,n\right\}$, we have a gadget in this multigraph. This gadget consists of two parallel edges, both connecting the source and the sink nodes ($s_i$ and $t_i$, respectively) of the particular player. One of these edges (referred to as the \textit{left} edge) has a capacity of $2^{\frac{n}{2}+1}$ whereas the other one also has at least this capacity. The gadgets are arranged in a row such that $s_i=t_{i+1}$ for $i\in\left\{\frac{n}{2}+1,\dots,n-1\right\}$. All the other players $i\in\left\{1,\dots,\frac{n}{2}\right\}$ have one disjoint source node $s_i$ each and $t_n$ as sink node. Moreover, there is one edge connecting $s_i$ and $s_{\frac{n}{2}}$ with capacity $2^{i-1}$. This results in a network as shown in Figure~\ref{fig_thm_china_falsch}. Obviously, the number of edges is in $\Theta\left(n\right)$.

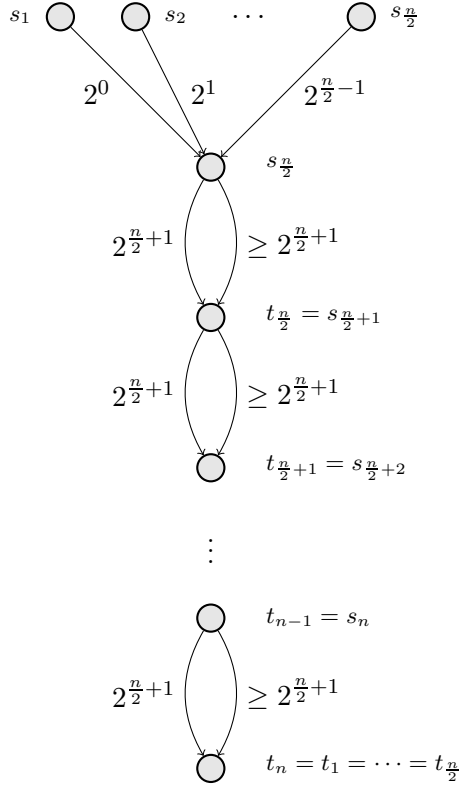
\begin{figure}
\centering
\begin{tikzpicture}[scale=2]
\node[draw,circle,minimum width=0.08cm,thick,fill=black!10] (a) at (-1,1) {};
\node[draw,circle,minimum width=0.08cm,thick,fill=black!10] (b) at (-0.5,1) {};
\node at (0.25,1) {$\hdots$};
\node[draw,circle,minimum width=0.08cm,thick,fill=black!10] (c) at (1,1) {};
\foreach \y/\name in {0/$s_\frac{n}{2}$,1/$t_\frac{n}{2}=s_{\frac{n}{2}+1}$,2/$t_{\frac{n}{2}+1}=s_{\frac{n}{2}+2}$,3/$t_{n-1}=s_n$,4/$t_n=t_1=\dots=t_\frac{n}{2}$}{
\node[draw,circle,minimum width=0.08cm,thick,fill=black!10] (\y) at (0,-\y) {};
\node[right of=\y,anchor=west,xshift=-.4cm] (a\y) {\footnotesize{\name}};
}
\node (8) at (0,-2.5) {$\vdots$};
\foreach \x/\y in {0/1,1/2,3/4}{
\draw[->] (\x) to [bend left]  node[right,anchor=west] {$\geq2^{\frac{n}{2}+1}$} (\y);
\draw[->] (\x) to [bend right]  node[left,anchor=east] {$2^{\frac{n}{2}+1}$} (\y);
}
\draw[->] (a) to  node[left,xshift=-.2cm,anchor=east] {$2^0$} (0);
\draw[->] (b) to  node[right,xshift=.1cm,anchor=west] {$2^1$} (0);
\draw[->] (c) to  node[right,xshift=.1cm,anchor=west] {$2^{\frac{n}{2}-1}$} (0);
\node[left of=a,anchor=east,xshift=.75cm] (a) {\footnotesize{$s_1$}};
\node[right of=b,anchor=west,xshift=-.75cm] (b) {\footnotesize{$s_2$}};
\node[right of=c,anchor=west,xshift=-.75cm] (c) {\footnotesize{$s_\frac{n}{2}$}};
\end{tikzpicture}
\caption{Illustration of the network in the proof of Theorem~\ref{thm_china_falsch}.}
\label{fig_thm_china_falsch}
\end{figure}

Note that, independently of the path a player $i\in\left\{1,\dots,\frac{n}{2}\right\}$ chooses, he is always assigned the respective bandwidth $2^{i-1}$ in the max-min fair allocation. This is because the residual capacity of an edge from the gadgets is larger than each of the bandwidth of players from $\left\{1,\dots,\frac{n}{2}\right\}$, even if all these players use this edge.

Consequently, the players $\left\{1,\dots,\frac{n}{2}\right\}$ are capable of choosing any natural number between $2^{\frac{n}{2}}+1$ and $2^{\frac{n}{2}+1}$ for the residual capacity of each of the $\frac{n}{2}$ left edges in the different gadgets. More specifically, let $x_{\frac{n}{2}}x_{\frac{n}{2}-1}\dots x_{1}$ be the binary representation of a natural number $x$ such that $2^{\frac{n}{2}+1}-x$ is from that interval. To obtain the desired residual capacity on a left edge in a given gadget, player $i$ simply chooses this edge in his path if and only if we have $x_{i}=1$. This has indeed the desired effect since
$2^{\frac{n}{2}+1}-\sum_{i\in N: x_i=1}2^{i-1}=2^{\frac{n}{2}+1}-x.$
We now give a lower bound on the number of different ordered allocation vectors. Since we want to derive a lower bound, it suffices to show the claim for allocations where the residual capacity of the left edge in the $i$-th gadget (i.e. the one of player $\frac{n}{2}+i$) is between
$
 2^{\frac{n}{2}}+1+\left(i-1\right)\cdot\left\lfloor(2^{\frac{n}{2}}-1)/n\right\rfloor$ and $2^{\frac{n}{2}}+1+i\cdot\left\lfloor(2^{\frac{n}{2}}-1)/n\right\rfloor$
and player $i$ chooses this edge. In these allocations, the bandwidth of player $i$ occurs in the ordered allocation vector \textit{before} the one of player $i+1$, for all $i\in\left\{\frac{n}{2}+1,\dots,n-1\right\}$. Consequently, the claim is implied by the following bound on the number of ordered allocations 
\[
 \left\lfloor\frac{2^{\frac{n}{2}}-1}{n}\right\rfloor^\frac{n}{2}=\left(\frac{2^{\Omega\left(n\right)}}{2^{\mathcal{O}\left(\log n\right)}}\right)^{\Omega\left(n\right)}=2^{\Omega\left(n\right)\cdot \Omega\left(n\right)}=2^{\Omega\left(n^2\right)}.
\]
\end{proof}

We now provide an \textit{upper} bound on the number of ordered values of the potential, even for general progressive filling games. For $m = \Theta(n)$, this yields an upper bound of $2^{O(n^2)}$.

\begin{theorem}
	\label{thm_impseq} 
	For arbitrary PFGs with the potential function $\phi$, it holds that 
	$\left|\Phi_{\ord}\right|\leq2^{n^2}\cdot m^n$.
\end{theorem}

\begin{proof}
Let $\mathcal{G}$ be a PFG with  potential function $\phi$. We claim that the number of different vectors up to the $k$-th position (for $k\leq n$) in $\left|\Phi_{\ord}\right|$ is at most $2^{k\cdot n}\cdot m^k$. It is shown via induction on $k$.
 
 For $k=0$, the claim is clear as there is only the vector of dimension $0$. So let $n\geq k>0$ and assume there are at most $2^{\left(k-1\right)\cdot n}\cdot m^{\left(k-1\right)}$ different vectors up to the position $k-1$ in $\Phi_{\ord}$. We now fix the first $k-1$ positions of a vector in $\Phi_{\ord}$ and bound the number of entries at the $k$-th position. Note that one can calculate the next finishing time given the resource which is saturated and the subset of players on that resource. Since there are $2^n\cdot m$ such combinations, the claim follows. 
\end{proof}

Theorem~\ref{SE_hard} shows that computing SE is NP-hard in MMFGs. For general PFGs with constant allocation rates (i.e., weighted MMF allocations), the same result holds even for single-commodity network games with two players. Hence, extending Dual Greedy to compute SE in polynomial time for this case is impossible.

\begin{theorem}
	\label{thm_sehard2}
	Let $v_1\neq v_2$ be two constant allocation rate functions and consider the class of single-commodity network PFGs with two players and $v_1,v_2$ as allocation rate functions. In this class, the computation of SE is NP-hard.
\end{theorem}

\begin{proof}
We reduce from the 2-directed-arc-disjoint-paths problem (2DADP). Let $\mathcal{I}$ be an instance of this problem. W.l.o.g., we can assume that there are paths in $D$ from $s_1$ to $t_1$ and from $s_2$ to $t_2$ and, further, that $v_1\equiv 1$ and $v_2\equiv \lambda$ where $\lambda<1$. Furthermore, we choose an $\varepsilon\in\left(0,1-\lambda\right)$.

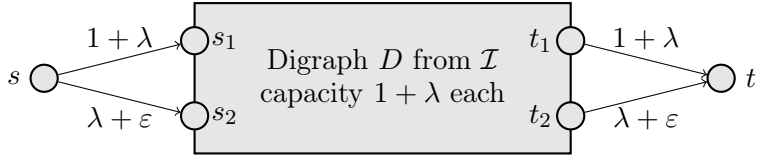
\begin{figure}[h]
\centering
\begin{tikzpicture}
\node[draw,minimum width=5cm,minimum height=2cm,thick,fill=black!10] (a) at (0,0) {\shortstack{Digraph $D$ from $\mathcal{I}$\\capacity $1+\lambda$ each}};
\node[draw,circle,minimum width=0.08cm,thick,fill=black!10] (s1) at (-2.5,0.5) {};
\node[draw,circle,minimum width=0.08cm,thick,fill=black!10] (s2) at (-2.5,-0.5) {};
\node[draw,circle,minimum width=0.08cm,thick,fill=black!10] (t1) at (2.5,0.5) {};
\node[draw,circle,minimum width=0.08cm,thick,fill=black!10] (t2) at (2.5,-0.5) {};
\node[draw,circle,minimum width=0.08cm,thick,fill=black!10] (s) at (-4.5,0) {};\node[draw,circle,minimum width=0.08cm,thick,fill=black!10] (t) at (4.5,0) {};
\draw[->] (s) to node[above] {$1+\lambda$} (s1);
\draw[->] (s) to node[below] {$\lambda+\varepsilon$} (s2);
\draw[->] (t1) to node[above] {$1+\lambda$} (t);
\draw[->] (t2) to node[below] {$\lambda+\varepsilon$} (t);
\node[left of=s,xshift=.6cm] {$s$};
\node[right of=s1,xshift=-.6cm] {$s_1$};
\node[right of=s2,xshift=-.6cm] {$s_2$};
\node[right of=t,xshift=-.6cm] {$t$};
\node[left of=t1,xshift=.6cm] {$t_1$};
\node[left of=t2,xshift=.6cm] {$t_2$};
\end{tikzpicture}
\caption{Illustration of the network of the game $\mathcal{G_I}$ in the proof of Theorem~\ref{thm_sehard2}.}
\label{fig_thm_sehard2}
\end{figure}

We construct the network underlying the single-commodity network PFG $\mathcal{G_I}$ by keeping $D$, adding the source and sink nodes $s$ and $t$, respectively, and the four edges $\left(s,s_1\right)$, $\left(s,s_2\right)$, $\left(t_1,t\right)$ and $\left(t_2,t\right)$. The capacities are 
\begin{align*}
c\left(s,s_1\right)=c\left(t_1,t\right)=1+\lambda\text{ and }c\left(s,s_2\right)=c\left(t_2,t\right)=\lambda+\varepsilon.
\end{align*}
The capacities of all edges occurring in $D$ are set to $1+\lambda$. This construction is illustrated in Figure~\ref{fig_thm_sehard2}.

We first show that each SE $S$ with two arc-disjoint paths from $s$ to $t$ certifies that $\mathcal{I}$ is solvable. To see this, note that player $1$ will always choose a path of the form $\left(s,s_1,\dots,t_1,t\right)$ because, even if he has to share an edge with player $2$, he gets bandwidth $1>\lambda+\varepsilon$. Hence, the path of player $2$ must indeed connect $s_2$ and $t_2$.

We also show that each SE $S$ without this property certifies that no arc-disjoint paths from $s_1$ to $t_1$ and from $s_2$ to $t_2$ exist. In $S$, both players share a common edge, i.e., player $1$ gets a bandwidth of $1$ and player $2$ a bandwidth of $\lambda$. Thus, if there were two arc-disjoint paths $\left(s,s_1,\dots,t_1,t\right)$ and $\left(s,s_2,\dots,t_2,t\right)$, both players could profitably change to these paths and get bandwidths of $1+\lambda$ and $\lambda+\varepsilon$. 
\end{proof}

Let us instead consider PNE, which may be easier to compute than SE. Similar to a result from~\cite{YangXFMZ10} for MMFGs, we first show that one can efficiently compute a unilateral improvement step for a given player in a PFG with constant allocation rate functions (if it exists). Using Theorem~\ref{thm_impseq}, computation of PNE can be done efficiently for a constant number of players.

\begin{lemma}
	\label{lem_impstep}
	In PFGs with constant allocation rate functions, an improving move of any player $i$ in any state $S$ can be 
	computed in polynomial time if it exists.
\end{lemma}
 
\begin{proof}
The bandwidth of player $i$ in the state $\left(\left\{r\right\},S_{-i}\right)$ can by computed in polynomial time by Algorithm~\ref{alg_pf}. Further, for a given strategy $S^\prime_i$, the bandwidth of player $1$ only depends on the resource which gets saturated first, i.e., 
$b_i\left(S^\prime_i,S_{-i}\right)=\min_{r\in S^\prime_i}b_i\left(\left\{r\right\},S_{-i}\right)$, 
which can easily be verified on Algorithm~\ref{alg_pf}. Thus, it suffices to calculate $\min_{r\in S^\prime_i}b_i\left(\left\{r\right\},S_{-i}\right)$ for all possible alternative strategies $S^\prime_i$ to decide whether there is an improvement step from $S$ for player $i$.

If the strategies are given explicitly as input, this value can be explicitly computed for each of the strategies. If strategies are given implicitly in the form of a network, we can use, e.g., Dijkstra's algorithm to find a path $P^\star$ with the maximum $\min_{r\in P^\star}b_i\left(\left\{r\right\},S_{-i}\right)$. 
\end{proof}

\begin{corollary}
	\label{thm_ne_poly}
  A PNE can be computed in polynomial time in PFGs with constant allocation rate functions and a constant number of 
	players.
\end{corollary}

\subsection{Quality of Equilibria}
In this section we prove results on PoA and PoS for NE in PFGs. In general, Theorem~\ref{thm_pos_nmmfgs} in the previous section yields a lower bound on the PoA in MMFGs of $\Omega\left(n\right)$. In fact, $n$ is also the correct upper bound on the PoA, for every PFG.

\begin{theorem}
 The PoA in PFGs is at most $n$.
\label{thm_poan}
\end{theorem}

\begin{proof}
 Consider a PFG $\mathcal{G}$ with state set $\mathcal{S}=\mathcal{S}_1\times\dots\times\mathcal{S}_n$. Then choose
 \[S_{j}\in\argmax_{S\in \mathcal{S}_1\cup\dots\cup \mathcal{S}_n}\min_{r\in S}c_r\]
 as a strategy with the maximum bottleneck resource, where $j$ is a player with $S_{j}\in\mathcal{S}_{j}$. Now consider an arbitrary NE $S^\prime$ and distinguish two cases.
 
 \begin{description}
 	\item[Case 1:] There is a saturated resource among the resources of $S_{j}$ in $S^\prime$. By the choice of $S_{j}$,  	\[\SW_\mathcal{G}\left(S^\prime\right) \quad \geq \quad \min_{r\in S_{j}}c_r \quad \geq \quad \frac{1}{n}\SW_\mathcal{G}\left(S^\star\right)\enspace,\]
 	for any other (optimal) state $S^\star$.
 	\item[Case 2:] There is no saturated resource among the resources of $S_{j}$ in $S^\prime$ but adding a bandwidth of $\delta>0$ on $S_{j}$ would saturate a resource. If player $j$ already uses a resource from $S_{j}$ in $S_j^\prime$, then $S^\prime$ is no NE because using $S_{j}$ as strategy instead would increase the bandwidth of player $j$ by $\delta$. 
 	
 	So let $S^\prime_{j}\cap S_{j}=\emptyset$. Then player $j$ must get a bandwidth of at least $\delta$ in $S^\prime$ since it would be profitable to use $S_{j}$ as strategy instead. Consequently, we get
 	\[\SW_\mathcal{G}\left(S^\prime\right)\geq\left(\min_{r\in S_j}c_r-\delta\right)+\delta=\min_{r\in S_j}c_r\geq\frac{1}{n}\SW_\mathcal{G}\left(S^\star\right)\]
for any other (optimal) state $S^\star$.
 \end{description}
\end{proof}

An improved result can be obtained for singleton games. We have already seen a lower bound of $2-\frac{1}{n}$ on the PoA in Theorem~\ref{thm_dg_approx}, as the tightness construction is a symmetric singleton MMFG. We now prove that this lower bound is tight, even in general singleton PFGs.

\begin{theorem}
 The PoA in singleton PFGs is $2-\frac{1}{n}$ and this bound is tight.
\label{thm_poa_singleton}
\end{theorem}
\begin{proof}
For a state $S\in\mathcal{S}$, we denote the resources used in $S$ by $R_S=\left\{r\in R\mid l_r\left(S\right)\neq 0\right\}$. As the limits of the indefinite integrals of the allocation rate functions for $x\rightarrow\infty$ are also $\infty$, we get that
\[\SW_{\mathcal{G}}\left(S\right)=\sum_{r\in R_S}c_r.\]

Further, let $S$ be a NE in $\mathcal{G}$ and $S^\star$ an arbitrary other (optimal) state. Since no player unilaterally deviates from his strategy in $S$ to a resource $r\notin R_S$, we must have
\begin{equation}
c_r \quad \leq \quad \min_{i\in N}b_i\left(S\right) \quad \leq \quad \frac{\SW_{\mathcal{G}}\left(S\right)}{n}\enspace,
\label{thm_poa_pfgs_eq1} 
\end{equation}
for all $r\notin R_S$. In particular, this holds for each resource $r\in R_{S^\star}\setminus R_S$. Distinguish two cases:

\begin{description}
	\item[Case 1:] We have $R_S\cap R_{S^\star}=\emptyset$. Then, by the previous observation, it follows that
\[\frac{\SW_{\mathcal{G}}\left(S^\star\right)}{\SW_{\mathcal{G}}\left(S\right)} \quad = \quad \frac{\sum_{r\in R_{S^\prime}\setminus R_S}c_r}{\sum_{r\in R_S}c_r} \quad \underset{(\ref{thm_poa_pfgs_eq1})}{\leq} \quad \frac{\SW_{\mathcal{G}}\left(S\right)}{\SW_{\mathcal{G}}\left(S\right)} \quad = \quad 1\enspace.\]
	\item[Case 2:] We have $R_S\cap R_{S^\star}\neq\emptyset$. Then, again by the previous observation, it follows that
\begin{align*}
 \frac{\SW_{\mathcal{G}}\left(S^\star\right)}{\SW_{\mathcal{G}}\left(S\right)} &\leq \frac{\sum_{r\in R_S}c_r+\sum_{r\in R_{S^\star}\setminus R_S}c_r}{\sum_{r\in R_S}c_r}\\
&\underset{(\ref{thm_poa_pfgs_eq1})}{\leq} \frac{\left(1+\frac{n-1}{n}\right)\cdot\SW_{\mathcal{G}}\left(S\right)}{\SW_{\mathcal{G}}\left(S\right)}\\
&=2-\frac{1}{n}\enspace.
\end{align*}
\end{description} 
\end{proof}

Finally, we show a lower bound of $n$ on the PoS if we leave the singleton case. For multi-commodity network MMFGs we proved a bound of $\Omega(n)$ in Theorem~\ref{thm_pos_nmmfgs} in the previous section. We show that this lower bound can even be established in single-commodity network PFGs. The reason for this is that a player with a fast-growing bandwidth may make his decision (nearly) unaffected of the decisions of all the other player decisions and hence possibly blocks all strategies for other players. This argument even applies if we only allow constant allocation rate functions.

\begin{theorem}
 The PoS in single-commodity network PFGs with constant allocation rate functions is at least $n$.
\label{thm_pos_scnpfg}
\end{theorem}
\begin{proof}
 For each $\varepsilon\in\left(0,1\right]$, we construct a family of single-commodity PFGs $\mathcal{G}_{n,\varepsilon}$ with $n$ players and PoS at least $\frac{n}{1+2\varepsilon}$. For such a game, we employ a gadget $G_{n,i}$ from the proof of Theorem~\ref{thm_spoakn} as underlying network. We omit the $i$ in the indices and call $u_i$ and $u_{i+1}$ simply $s$ and $t$, respectively. Further, we adapt the capacities of the edges in the following way. We set
\[c_e=\begin{cases}
       1+\varepsilon,&\text{if $e=\left(s,v_1\right)$ or $e=\left(w_n,t\right)$ or $e$ is not incident to $s$ or $t$} \\
       1,&\text{else}
      \end{cases}\enspace,\]
for all $e\in E$. To obtain a PFG from this network, we equip player $1$ with an allocation rate function which is constantly $1$ and all the other players from $\left\{1,\dots,n\right\}$ with functions which are constantly $\frac{\varepsilon}{n}$.

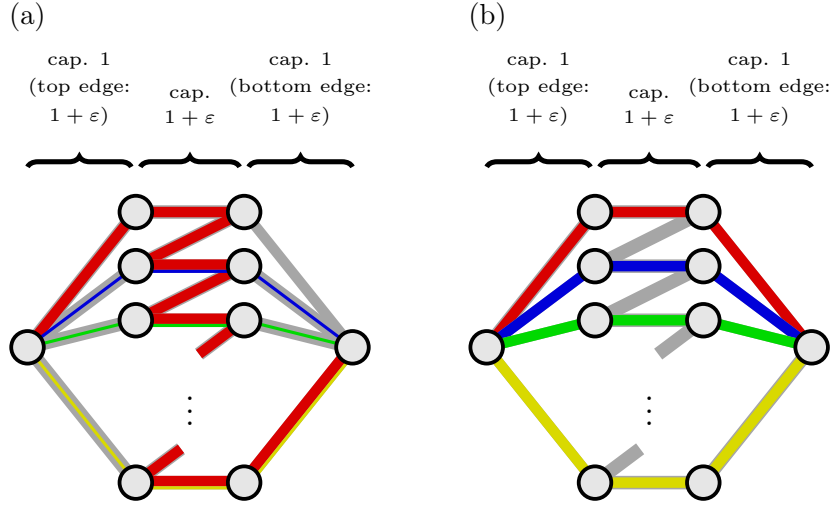
\begin{figure}
\centering
\subfiguretopcaptrue
\subfigure[][]{\scalebox{1.2}{\begin{tikzpicture}
\definecolor{darkred}{rgb}{.85,0,0}
\definecolor{darkblue}{rgb}{0,0,.85}
\definecolor{darkgreen}{rgb}{0,.85,0}
\definecolor{darkyellow}{rgb}{.85,.85,0}
\draw [decoration={brace,amplitude=0.3em},decorate,ultra thick] (0,2) -- node[above,yshift=.3cm] {\tiny{\shortstack{cap. $1$\\(top edge:\\$1+\varepsilon$)}}} (1.15,2);
\draw [decoration={brace,amplitude=0.3em},decorate,ultra thick] (1.25,2) -- node[above,yshift=.3cm] {\tiny{\shortstack{cap.\\$1+\varepsilon$}}} (2.35,2);
\draw [decoration={brace,amplitude=0.3em},decorate,ultra thick] (2.45,2) -- node[above,yshift=.3cm] {\tiny{\shortstack{cap. $1$\\(bottom edge:\\$1+\varepsilon$)}}} (3.6,2);
\draw[-,line width=0.15cm,draw=black!35] (0,0) -- (1.2,0.9);
\draw[-,line width=0.15cm,draw=black!35] (0,0) -- (1.2,0.3);
\draw[-,line width=0.03461538461538461cm,draw=darkgreen] (0,0) -- (1.2,0.3);
\draw[-,line width=0.15cm,draw=black!35] (0,0) -- (1.2,-1.5);
\draw[-,line width=0.03461538461538461cm,draw=darkyellow] (0,0) -- (1.2,-1.5);
\draw[-,line width=0.15cm,draw=black!35] (2.4,0.9) -- (3.6,0);
\draw[-,line width=0.15cm,draw=black!35] (2.4,0.3) -- (3.6,0);
\draw[-,line width=0.03461538461538461cm,draw=darkgreen] (2.4,0.3) -- (3.6,0);
\draw[-,line width=0.15cm,draw=black!35] (2.4,-1.5) -- (3.6,0);
\draw[-,line width=0.11538461538461536cm,draw=darkred] (2.386484962913486,-1.4891879703307889) -- (3.5864849629134863,0.010812029669211186);
\draw[-,line width=0.03461538461538461cm,draw=darkyellow] (2.445050123621713,-1.5360400988973706) -- (3.645050123621713,-0.03604009889737063);
\draw[-,line width=0.15cm,draw=black!35] (0,0) -- (1.2,1.5);
\draw[-,line width=0.11538461538461536cm,draw=darkred] (0,0) -- (1.2,1.5);
\draw[-,line width=0.03461538461538461cm,draw=darkblue] (0,0) -- (1.2,0.9);
\draw[-,line width=0.11538461538461536cm,draw=black!35] (2.4,1.5) -- (3.6,0);
\draw[-,line width=0.15cm,draw=black!35] (1.2,1.5) -- (2.4,1.5);
\draw[-,line width=0.15cm,draw=black!35] (2.4,1.5) -- (1.2,0.9);
\draw[-,line width=0.11538461538461536cm,draw=darkred] (2.4,1.5) -- (1.2,0.9);
\draw[-,line width=0.15cm,draw=black!35] (1.2,0.9) -- (2.4,0.9);
\draw[-,line width=0.15cm,draw=black!35] (2.4,0.9) -- (1.2,0.3);
\draw[-,line width=0.11538461538461536cm,draw=darkred] (2.4,0.9) -- (1.2,0.3);
\draw[-,line width=0.15cm,draw=black!35] (2.4,0.3) -- ++(-0.5,-0.375);
\draw[-,line width=0.11538461538461536cm,draw=darkred] (2.4,0.3) -- ++(-0.5,-0.375);
\draw[-,line width=0.15cm,draw=black!35] (1.2,-1.5) -- ++(0.5,0.375);
\draw[-,line width=0.11538461538461536cm,draw=darkred] (1.2,-1.5) -- ++(0.5,0.375);
\draw[-,line width=0.11538461538461536cm,draw=darkred] (1.2,1.5) -- (2.4,1.5);
\draw[-,line width=0.11538461538461536cm,draw=darkred] (1.2,0.9173076923076924) -- (2.4,0.9173076923076924);
\draw[-,line width=0.03461538461538461cm,draw=darkblue] (1.2,0.8423076923076923) -- (2.4,0.8423076923076923);
\draw[-,line width=0.15cm,draw=black!35] (1.2,0.3) -- (2.4,0.3);
\draw[-,line width=0.11538461538461536cm,draw=darkred] (1.2,0.3173076923076923) -- (2.4,0.3173076923076923);
\draw[-,line width=0.03461538461538461cm,draw=darkgreen] (1.2,0.2423076923076923) -- (2.4,0.2423076923076923);
\draw[-,line width=0.15cm,draw=black!35] (1.2,-1.5) -- (2.4,-1.5);
\draw[-,line width=0.11538461538461536cm,draw=darkred] (1.2,-1.4826923076923078) -- (2.4,-1.4826923076923078);
\draw[-,line width=0.03461538461538461cm,draw=darkyellow] (1.2,-1.5576923076923077) -- (2.4,-1.5576923076923077);
\draw[-,line width=0.03461538461538461cm,draw=darkblue] (2.4,0.9) -- (3.6,0);

\node[draw,circle,minimum width=0.12cm,very thick,fill=black!10] (0) at (0,0) {};
\node[draw,circle,minimum width=0.12cm,very thick,fill=black!10] (1) at (1.2,1.5) {};
\node[draw,circle,minimum width=0.12cm,very thick,fill=black!10] (2) at (2.4,1.5) {};
\node[draw,circle,minimum width=0.12cm,very thick,fill=black!10] (3) at (1.2,0.9) {};
\node[draw,circle,minimum width=0.12cm,very thick,fill=black!10] (4) at (2.4,0.9) {};
\node[draw,circle,minimum width=0.12cm,very thick,fill=black!10] (5) at (1.2,0.3) {};
\node[draw,circle,minimum width=0.12cm,very thick,fill=black!10] (6) at (2.4,0.3) {};
\node[draw,circle,minimum width=0.12cm,very thick,fill=black!10] (7) at (1.2,-1.5) {};
\node[draw,circle,minimum width=0.12cm,very thick,fill=black!10] (8) at (2.4,-1.5) {};
\node[draw,circle,minimum width=0.12cm,very thick,fill=black!10] (9) at (3.6,0) {};
\node at (1.8,-0.6) {$\vdots$};
\end{tikzpicture}}}\hspace{1cm}
\subfigure[][]{\scalebox{1.2}{\begin{tikzpicture}
\definecolor{darkred}{rgb}{.85,0,0}
\definecolor{darkblue}{rgb}{0,0,.85}
\definecolor{darkgreen}{rgb}{0,.85,0}
\definecolor{darkyellow}{rgb}{.85,.85,0}
\draw [decoration={brace,amplitude=0.3em},decorate,ultra thick] (0,2) -- node[above,yshift=.3cm] {\tiny{\shortstack{cap. $1$\\(top edge:\\$1+\varepsilon$)}}} (1.15,2);
\draw [decoration={brace,amplitude=0.3em},decorate,ultra thick] (1.25,2) -- node[above,yshift=.3cm] {\tiny{\shortstack{cap.\\$1+\varepsilon$}}} (2.35,2);
\draw [decoration={brace,amplitude=0.3em},decorate,ultra thick] (2.45,2) -- node[above,yshift=.3cm] {\tiny{\shortstack{cap. $1$\\(bottom edge:\\$1+\varepsilon$)}}} (3.6,2);

\draw[-,line width=0.15cm,draw=black!35] (0,0) -- (1.2,1.5);
\draw[-,line width=0.11538461538461536cm,draw=darkred] (0,0) -- (1.2,1.5);
\draw[-,line width=0.11538461538461536cm,draw=black!35] (0,0) -- (1.2,0.9);
\draw[-,line width=0.11538461538461536cm,draw=darkblue] (0,0) -- (1.2,0.9);
\draw[-,line width=0.11538461538461536cm,draw=black!35] (0,0) -- (1.2,0.3);
\draw[-,line width=0.11538461538461536cm,draw=darkgreen] (0,0) -- (1.2,0.3);
\draw[-,line width=0.11538461538461536cm,draw=black!35] (0,0) -- (1.2,-1.5);
\draw[-,line width=0.11538461538461536cm,draw=darkyellow] (0,0) -- (1.2,-1.5);
\draw[-,line width=0.11538461538461536cm,draw=black!35] (2.4,1.5) -- (3.6,0);
\draw[-,line width=0.11538461538461536cm,draw=darkred] (2.4,1.5) -- (3.6,0);
\draw[-,line width=0.11538461538461536cm,draw=black!35] (2.4,0.9) -- (3.6,0);
\draw[-,line width=0.11538461538461536cm,draw=darkblue] (2.4,0.9) -- (3.6,0);
\draw[-,line width=0.11538461538461536cm,draw=black!35] (2.4,0.3) -- (3.6,0);
\draw[-,line width=0.11538461538461536cm,draw=darkgreen] (2.4,0.3) -- (3.6,0);
\draw[-,line width=0.15cm,draw=black!35] (2.4,0.3) -- ++(-0.5,-0.375);
\draw[-,line width=0.15cm,draw=black!35] (1.2,-1.5) -- ++(0.5,0.375);
\draw[-,line width=0.15cm,draw=black!35] (2.4,0.9) -- (1.2,0.3);
\draw[-,line width=0.15cm,draw=black!35] (2.4,1.5) -- (1.2,0.9);
\draw[-,line width=0.15cm,draw=black!35] (2.4,-1.5) -- (3.6,0);
\draw[-,line width=0.11538461538461536cm,draw=darkyellow] (2.4,-1.5) -- (3.6,0);
\draw[-,line width=0.15cm,draw=black!35] (1.2,1.5) -- (2.4,1.5);
\draw[-,line width=0.11538461538461536cm,draw=darkred] (1.2,1.5) -- (2.4,1.5);
\draw[-,line width=0.15cm,draw=black!35] (1.2,0.9) -- (2.4,0.9);
\draw[-,line width=0.11538461538461536cm,draw=darkblue] (1.2,0.9) -- (2.4,0.9);
\draw[-,line width=0.15cm,draw=black!35] (1.2,0.3) -- (2.4,0.3);
\draw[-,line width=0.11538461538461536cm,draw=darkgreen] (1.2,0.3) -- (2.4,0.3);
\draw[-,line width=0.15cm,draw=black!35] (1.2,-1.5) -- (2.4,-1.5);
\draw[-,line width=0.11538461538461536cm,draw=darkyellow] (1.2,-1.5) -- (2.4,-1.5);

\node[draw,circle,minimum width=0.12cm,very thick,fill=black!10] (0) at (0,0) {};
\node[draw,circle,minimum width=0.12cm,very thick,fill=black!10] (1) at (1.2,1.5) {};
\node[draw,circle,minimum width=0.12cm,very thick,fill=black!10] (2) at (2.4,1.5) {};
\node[draw,circle,minimum width=0.12cm,very thick,fill=black!10] (3) at (1.2,0.9) {};
\node[draw,circle,minimum width=0.12cm,very thick,fill=black!10] (4) at (2.4,0.9) {};
\node[draw,circle,minimum width=0.12cm,very thick,fill=black!10] (5) at (1.2,0.3) {};
\node[draw,circle,minimum width=0.12cm,very thick,fill=black!10] (6) at (2.4,0.3) {};
\node[draw,circle,minimum width=0.12cm,very thick,fill=black!10] (7) at (1.2,-1.5) {};
\node[draw,circle,minimum width=0.12cm,very thick,fill=black!10] (8) at (2.4,-1.5) {};
\node[draw,circle,minimum width=0.12cm,very thick,fill=black!10] (9) at (3.6,0) {};
\node at (1.8,-0.6) {$\vdots$};
\end{tikzpicture}}}
\caption{Illustration of the network of the game $\mathcal{G}_{n,\varepsilon}$ from the proof of Theorem~\ref{thm_pos_scnpfg} with (a) the best NE $S^\prime$ and (b) the optimal state $S$ in terms of social welfare.}
\label{fig_pos_pfgs}
\end{figure}

Consider a state $S$ with social welfare $n$. Such a state evolves if player $i$ chooses the path $\left(s,v_i,w_i,t\right)$, for all $i$. Player $1$, however, has an incentive to use the path $\left(s,v_1,w_1,v_2\dots,v_n,w_n,t\right)$ instead of any other path with capacity $1$ -- even if player $1$ had to share a $\left(1+\varepsilon\right)$-edge with all the other players he would get a bandwidth larger than $1$. More specifically, if this edge is the first one saturated by Algorithm~\ref{alg_pf} (which is the case when player $1$ chooses the considered path), the finishing time of this edge is larger than $1$.

Now let $S^\prime$ be the NE with the highest social welfare. By the previous considerations, player $1$ uses the path $\left(s,v_1,w_1,v_2\dots,v_n,w_n,t\right)$ and obviously gets at most a bandwidth of $1+\varepsilon$ in $S^\prime$. Since we chose $\varepsilon\leq1$, all the other players get at most a bandwidth of $\frac{2\varepsilon}{n}$. This bound is tight if $\varepsilon=1$ and all these players only share an edge with player $1$.

We can now compare the social welfare of the states $S$ and $S^\prime$ (both illustrated in Figure~\ref{fig_pos_pfgs}) and obtain that $\PoS\left(\mathcal{G}_n\right)\geq\frac{n}{1+2\varepsilon}$.
Hence, we get $\sup\left\{\PoS\left(\mathcal{G}_{n,\varepsilon}\right)\mid n\in\mathbb{N}\wedge\varepsilon\in\mathbb{R}\wedge\varepsilon>0\right\}\geq n$, which implies the claim.
\end{proof}

\subsection{Changing the Allocation Rate Functions}
\label{sec:change}



Dual Greedy computes a SE that is a $2-\frac{1}{n}$-approximation and this bound is tight. To stabilize better solutions, in this section we take a ``protocol design'' approach. We assume the waterfilling algorithm can determine a set of constant allocation rate functions for each instance. Interestingly, for any given collection of players, resources, capacities and strategy sets, one can give constant allocation rate functions such that the resulting PFG has an SE with social welfare as high as the optimal value of the MCAP.

\begin{theorem}
	\label{thm_change}
	Let $\mathcal{G}$ be a PFG with player set $N$ and $v^\star$ be the optimal value of the MCAP. There are 
	constant allocation rate functions $\left(v'_i\right)_{i\in N}$ such that the maximal social welfare in 
	$\mathcal{G}$ with allocation rate functions replaced by $\left(v'_i\right)_{i\in N}$ is $v^\star$ and the SPoS in 
	this game is $1$.
\end{theorem}

\begin{proof}
Let the state $S=\left(S_1,\dots,S_n\right)$ along with the allocation $a=\left(a_1,\dots,a_n\right)$ be an optimal solution of the MCAP. We use allocation rate function $v'_i \equiv a_i$ for each player $i\in N$. We call the corresponding PFG $\mathcal{G'}$. If we run the progressive filling algorithm in $S$ with $v'$, all finishing times are exactly $1$ and the allocation is exactly $a$.\\
We show that $S$ is a SE in $\mathcal{G}$. Towards this, suppose that there is a coalition profitably deviating from $S$ to $T$. Then, by Lemma~\ref{lemma_cpf}, the finishing times and thus bandwidths of all players in $N \setminus C$ remain identical in $T$ whereas the players from $C$ strictly improve. Consequently, we have constructed a solution of the MCAP on $\mathcal{M}$ with a higher social welfare -- a contradiction. 
\end{proof}

Not surprisingly, this approach is intractable, as the MCAP is NP-hard to approximate to within a factor $\frac{3}{2}-\varepsilon$, even for arbitrary fixed rates.

\begin{theorem}
	\label{thm_mcap32}
	For $2$ players, it is NP-hard to approximate the MCAP with a factor of smaller than $\frac{3}{2}$. This also 
	holds for the MCAP with arbitrary fixed rates.
\end{theorem}

\begin{proof}
We use the reduction from the proof of Theorem~\ref{thm_sehard2} for $\lambda=1$ and $\varepsilon=0$. Since for two players, the allocation rate functions do not affect the social welfare in a given state, our argumentation works completely without allocation rate functions, even for the MCAP with fixed rates.

Each state with optimal social welfare $2$ certifies that there are no arc-disjoint paths $\left(v_1,\dots,v_k\right)$ and $\left(w_1,\dots,w_l\right)$ from $s_1$ to $t_1$ and from $s_2$ to $t_2$, respectively, since, otherwise, the state \linebreak $\left(\left(s_1,v_1,\dots,v_k,t_1\right), \left(s_2,w_1,\dots,w_l,t_2\right)\right)$ would attain a social welfare of $3$. 

Conversely, each state with a social welfare higher than $2$ must use two arc-disjoint paths in $D$. Further, these paths must obviously connect $s_1$ to $t_1$ and $s_2$ to $t_2$. Consequently, such a state certifies that the instance of 2DADP is solvable. 
\end{proof}

This implies that the approximation guarantee of Dual Greedy is optimal for $n=2$, even without requiring the output to be a SE. The idea behind the previous theorem extends also to approximate solutions of the MCAP. For the MCAP on single-commodity networks, a better $\frac 32$-approximation exists for $n=3$~\cite{BaierKS05} and can be obtained as follows: Run the maximum capacity augmenting path algorithm~\cite{AhujaMO93} on the given network for two iterations and decompose~\cite{AhujaMO93} the obtained flow into three paths (plus a circulation).
We use this approach to calculate an \textit{equilibrium state} that is a better approximation than the one calculated by Dual Greedy. By Theorem~\ref{thm_dg_approx}, this is not possible if the allocation rate functions are fixed, even for uniform ones. Adjusting allocation rate functions subject to the instance, however, allows to beat Dual Greedy, at least for $n=3$ and PNE.

\begin{theorem}
	\label{thm_improving_dg}
	In single-commodity networks with $3$ players, there exist constant allocation rate functions and a PNE that is a $
	\frac{3}{2}$-approximation to the MCAP. The allocation rate functions and the PNE can be computed in polynomial 
	time.
\end{theorem}

\begin{proof}
Let $S=\left(S_1,S_2,S_3\right)$ and the allocation $a=\left(a_1,a_2,a_3\right)$ represent the a $\frac{3}{2}$-approximation of the MCAP. As allocation rate function for player $i$, for all $i\in N$, we use the function $v_i\equiv a_i$. Note that, if we run Algorithm~\ref{alg_pf} on $S$, the finishing time is $1$, for each of the players, and $a$ is exactly the computed allocation.
%
We now invoke best-response dynamics starting from $S$ and iteratively compute and apply unilateral player deviations. By Corollary~\ref{thm_impseq}, this procedure can be implemented in polynomial time. We call the resulting state $S^\star$. Using Lemma~\ref{lemma_cpf}, we know that the finishing times of the players never sink below $1$ during that procedure. Consequently, $S^\star$ is at most a $\frac{3}{2}$-approximation to the MCAP. 
\end{proof}

Indeed, we can start with an arbitrary approximate solution of the MCAP, set the allocation rates such that finishing times are all 1, and then every unilateral (coalitional) improvement dynamics will lead to a PNE (SE) that only improves social welfare. Exploring this idea is a very interesting avenue for future work.



\section*{Acknowledgement}

We thank Berthold V\"ocking for helpful comments regarding the model underlying this paper. 

\bibliographystyle{abbrv}
\bibliography{INFOCOM14-Arxiv3.bbl}

\end{document}